\documentclass[12pt]{article}
\usepackage[a4paper]{geometry}
\usepackage[italian,english]{babel}

\usepackage{amsmath}
\usepackage{amsfonts}
\usepackage{amstext}
\usepackage{amssymb}
\usepackage{amsthm}
\usepackage{amscd}

\usepackage{hyperref}
\hypersetup{colorlinks,
linkcolor=myrefcolor,
citecolor=mycitecolor,
urlcolor=myurlcolor}
\usepackage[capitalize]{cleveref}
\usepackage{caption}
\usepackage{etaremune}
\usepackage{bibentry}

\usepackage{xcolor}
\definecolor{myurlcolor}{rgb}{0,0,0.4}
\definecolor{mycitecolor}{rgb}{0,0.5,0}
\definecolor{myrefcolor}{rgb}{0.5,0,0}
\usepackage{graphicx}
\usepackage{tikz}
\usepackage{tikz-cd}
\usepackage{mathrsfs}

\usepackage{etoolbox}
\usepackage{makeidx}
\usepackage{sectsty}
\usepackage{dsfont}
\usepackage{enumitem} 
\usepackage[]{latexsym}
\usepackage{braket}
\usepackage{caption}
\usepackage[utf8]{inputenx}
\usepackage[T1]{fontenc}
\usepackage{lmodern}
\usepackage{textcomp}
\usepackage{microtype}
\usepackage{totcount}
\usepackage{blindtext}

\newtheorem{theorem}{Theorem}

\newtheorem{proposition}{Proposition}

\newtheorem*{proof*}{Proof}

\newcommand{\be}{\begin{equation}}
\newcommand{\ee}{\end{equation}}
\newcommand{\bea}{\begin{eqnarray}}
\newcommand{\eea}{\end{eqnarray}}

\title{Schwinger's Picture of Quantum Mechanics III: The statistical interpretation}

\author{F. M. Ciaglia$^{1,5}$  \href{https://orcid.org/0000-0002-8987-1181}{\includegraphics[scale=0.7]{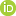}}, A. Ibort$^{2,3,6}$\href{https://orcid.org/0000-0002-0580-5858}{\includegraphics[scale=0.7]{ORCID.png}}, G. Marmo$^{4,7}$\href{https://orcid.org/0000-0003-2662-2193}{\includegraphics[scale=0.7]{ORCID.png}}\\
\footnotesize{$^{1}$\textit{ Max Planck Institute for Mathematics in the Sciences, Leipzig, Germany}} \\
\footnotesize{$^{2}$\textit{ ICMAT, Instituto de Ciencias Matem\'{a}ticas (CSIC-UAM-UC3M-UCM)}}  \\
\footnotesize{$^{3}$\textit{ Depto. de Matem\'aticas, Univ. Carlos III de Madrid, Legan\'es, Madrid, Spain}}  \\
\footnotesize{$^{4}$\textit{ Dipartimento di Fisica ``E. Pancini'', Universit\`a di Napoli Federico II, Napoli, Italy}} \\
\footnotesize{$^{5}$\textit{ e-mail: \texttt{florio.m.ciaglia[at]gmail.com}}}, \\
\footnotesize{ $^{6}$\textit{ e-mail: \texttt{albertoi[at]math.uc3m.es}}} \\ 
\footnotesize{$^{7}$\textit{ e-mail: \texttt{marmo[at]na.infn.it}}}}

\begin{document}

\maketitle

\begin{abstract}
Schwinger's algebra of selective measurements has a natural interpretation in the formalism of groupoids.
Its kinematical foundations, as well as the structure of the algebra of observables of the theory, was presented in \cite{Ib18, Ib18b}.  In this paper, a closer look to the statistical interpretation of the theory is taken and it is found that an interpretation in terms of Sorkin's quantum measure emerges naturally.  It is proven that a suitable class of states of the algebra of virtual transitions of the theory  allows to define quantum measures by means of the  corresponding decoherence functionals.   Quantum measures satisfying a reproducing property are described and a class of states, called factorizable states, possessing the Dirac-Feynman `exponential of the action' form are characterized.   Finally,  Schwinger's transformation functions are interpreted similarly as transition amplitudes defined by suitable states.   The simple examples of the qubit and the double slit experiment are described in detail,  illustrating the main aspects of the theory.
\end{abstract}

\tableofcontents

\section{Introduction: Groupoids and Quantum Mechanics} 

\subsection{On the statistical interpretation of Quantum Mechanics by Feynman and Schwinger}
A careful interpretation of the probabilistic nature of Quantum Mechanics led both J. Schwinger and R.P. Feynman to their own, quite disparate, formulations of the theory.    Already in his Ph. D. Thesis \cite{Fe05} and in the seminal article \cite{Fe48}, Feynman forcefully stated that \textit{``...it seems worthwhile to emphasize the fact that  [the observed experimental facts] are all simply direct consequences of Eq. (\ref{feynman}),  
\begin{equation}\label{feynman}
\varphi_{ab} = \sum_b \varphi_{ab} \varphi_{bc}\, ,
\end{equation}
for it is essentially this equation that is fundamental in my formulation of quantum mechanics''} \cite{Fe48}.  The quantum amplitudes $\varphi_{ab}$ being such that $|\varphi_{ab}|^2$ represent the classical probability that, if measurement $A$ gave the result $a$, then measurement $B$ will give the result $b$.    Then, Feynman, following Dirac's powerful insight \cite{Di33}, proceeded by postulating that this quantum probability amplitude ``has a phase proportional to the action'' and implemented his sum over histories description of quantum mechanics \cite{Fe48} that, as stated by Yourgrau and Mandelstan, \textit{``...cannot fail  but to observe that Feynman's principle -- and this is no hyperbole -- expresses the laws of quantum mechanics in an exemplary neat and elegant manner''} \cite[Footnote 6]{Yo68}.

Alternatively, J. Schwinger introduced the statistical interpretation of his selective measurement symbols by stating: \textit{``...measurements of properties $B$, performed on a system in a state $c'$ that refers to properties incompatible with $B$, will yield a statistical distribution of the possible values.  Hence, only a determinate fraction of the systems emerging from the first stage will be accepted by the second stage.  We express this by the general multiplication law:
$$
M(a',b') M(c',d') = \langle b' \mid c' \rangle M(a',d') \, ,
$$ 
where $ \langle b' \mid c' \rangle$ is a number characterising the \underline{statistical relation}\footnote{The underlying is ours.} bewteen states $b'$ and $c'$''} \cite[Chap. 1.3]{Sc91}.  We will just add here that Schwinger's transformation functions $ \langle b' \mid c' \rangle$ played an instrumental role in his development of quantum electrodynamics \cite{Sc51}.   

Much more recently, R. Sorkin, in his paper presenting a quantum measure interpretation of quantum mechanics \cite{So94}, when discussing the standard statistical interpretation of quantum mechanics, stressed: \textit{``...to the untutored mind, however, the formal rules of the path-integral scheme, could seem unnatural and contrived.  Why are probabilites squares of amplitudes...?''}.  

In this paper, we will try to offer a new approximation to all these ideas, tying together the (apparently) disparate statistical interpretations of quantum mechanics upon which Schwinger and Feynman founded their own descriptions of the theory, and putting them under the unifying conceptual framework provided by Sorkin's quantum measure interpretation of quantum mechanics.  
This will be achieved by using the recently proposed groupoid interpretation of Schwinger's algebra of measurements in \cite{Ib18, Ib18b}.    It will be shown that the obtained results reproduce nicely both Schwinger's and Feynman's interpretations, providing an explicit proof of Feynman's equation, the construction of Schwinger's transitions functions, and Sorkin decoherence functionals from first principles.  We must acknowledge that there are no fundamentally new results but, as Feynman's himself stated in the introduction to his epoch making paper \cite{Fe48}: \textit{``... there is a pleasure in recognizing old things from a new point of view. Also, there are problems for which the new point of view offers a distinct advantage''}.
 
In our previous works \cite{Ib18, Ib18b}, both the kinematical background and the basic dynamical structures for a new description of quantum mechanical systems inspired by Schwinger's algebra of selective measurement were presented.  
It was argued that the basic kinematical structure needed to describe a theory of quantum systems can be developed from the primary notions of \textit{outcomes} or \textit{events}, \textit{transitions} and \textit{transformations} that, from a mathematical viewpoint, satisfy the algebraic properties of a 2-groupoid.   `Outcomes' or `events'\footnote{The name was chosen for the lack of a better word.  `Event' has a precise meaning in probability and causality theories, and it collides with the meaning assigned by Sorkin to it. Schwinger's called them `states', however `state' will be used in a widely extended technical sense that, at the same time, captures perfectly well the required statistical meaning. Thus we will stick with `outcomes'  (or `events') for the time being.} and `transitions' provide a natural abstract setting for Schwinger's notion of physical selective measurements, and form a groupoid from the mathematical perspective.  

The concept of `outcome' extends the notion of `state' used by Schwinger (that in his case coincides with the standard notion of maximal compatible family of measurements): \textit{``... a complete measurement, which is such that the systems chosen possess definite values for the maximum number of attributes...  Thus the optimum state of knowledge concerning a given system is realized by subjecting it to a complete selective measurement}'' \cite[Chap. 1.2]{Sc91}.   We would like to extend such notion to consider possible outcomes of observations or manipulations performed on a given system, not necessarily complete in any sense, so that we may consider histories of outcomes as natural ingredients of the theory (a possibility already anticipated by Feynman: \textit{``....  Suppose a measurement is made which is capable only of determining that the path lies somewhere within a region $R$.  The measurement is to be what we might call an `ideal measurement'... I have not been able to find a precise definition'',} \cite{Fe48}).

On the other hand, in Schwinger's conceptualisation, the notion of `transition' is clearly identified and corresponds to `measurements that change the state' \cite[Chap. 1.3]{Sc91}.   Thus, the notion of transition introduced in \cite{Ib18} extends the notion of selective measurements that change the state to include all physical feasible changes between events of the system.  Transitions  can be naturally composed and their composition law satisfies the axioms of a groupoid which constitutes the fundamental algebraic structure of the theory.  

It is instrumental the assumption that transitions are invertible.   In this sense, we agree with Feynman when states quite forcefully:  \textit{``The fundamental (microscopic) phenomena in nature are symmetrical with respect to interchange of past and future''} \cite[Chap. I, p. 3]{Fe05}. We share this principle, that leads to the assumption that the transitions of the theory must be invertible, hence, define a groupoid and not just a category\footnote{However, Schwinger, even if his formalism implies that the selective measurements that change the state are `invertible', only reluctantly acknowledges that when states:  \textit{``...the arbitrary convention that accompanies the interpretation of the measurement symbols and their products - the order of events is read from right to left (sinistrally), but any measurement symbol equation is equally valid if interpreted in the opposite sense (dextrally), and no physical result should depend upon which convention is employed''} \cite[Chap. 1.7]{Sc91}.}.   

Passing from a `reference system' with outcomes denoted by $a$ and transitions $\alpha$ to another with events  $b$ and transitions $\beta$ requires a theory of transformations. Such theory is developed by Schwinger \cite[Chap. 2.5]{Sc91} and reproduces the standard theory of unitary operators developed by Dirac.  However, previous to that, and at a more basic level, Schwinger introduces the notion of \textit{transformation function} $\langle a' \mid b' \rangle$, a notion that will be discussed in the present paper:  \textit{``...measurement symbols of one type can be expressed as linear combinations of the measurement symbols of another type.  From its role in effecting such connections the totality of numbers $\langle a' \mid b' \rangle$ is called the transformation function relating the $a-$ and the $b-$descriptions''.}    In Schwinger's formulation, the transformation functions arise as a concrete representation of an abstract operation that was coined `transformation' and that affects  specific transitions as explained in \cite{Ib18}.   The theory of `transformations' thus developed fits naturally in the algebraic setting of the theory of groupoids and determines a 2-groupoid structure on top of Schwinger's groupoid of outcomes and transitions.    

The previous ideas fix the kinematical framework of the theory as discussed in \cite{Ib18}, where no attempt to introduce a dynamical content was made.  In this sense we were following R. Sorkin's dictum of \textit{``proposing a framework in which the ontology or `kinematics' and the dynamics or `laws of motion' are as sharply separated from each other as they are in classical physics'' }\cite{So94}.  

The dynamical aspects of the theory were discussed though in \cite{Ib18b}.  The departing point of the analysis was the key idea of considering an observable as the assignment of an amplitude to any transition, that is, an observable is a function on the groupoid of transitions, an idea which just reflects the abstraction of the determination of an observable by means of the amplitudes $\langle a | A | a' \rangle$.    The notion of observables thus introduced leads in a natural way to the construction of the $C^*$-algebra of observables of the system and a Heisenberg-like formulation of dynamics as infinitesimal generators of  their automorphisms. 

Physical states of the system correspond in this way to states of the $C^*$-algebra of observables, that is, normalised positive linear functionals on the algebra, opening the road to the interpretation of the theory in terms of Hilbert spaces and operators by applying the GNS construction associated to any state of the theory, an idea that will be used repeatedly in the present paper (another use of the theory is shown in \cite{Co19} where coherent states are nicely described in this setting).  

We must stress here that this approach departs from Schwinger's derivation of the laws of dynamics from a quantum dynamical principle, that nevertheless will be undertaken in a future publication \cite{Ib18c}.  We consider that the approach taken in \cite{Ib18b} is more natural and we agree with R. Sorkin: \textit{``...quantum theory differs from classical mechanics in its dynamics, which ... is stochastic rather than deterministic. As such the theory functions by furnishing probabilities for sets of histories''} \cite{So94}.   In this sense, the dynamical theory that we propose is closer in spirit to R. Sorkin's understanding of quantum mechanics as a quantum measure theory, point of view that will be one of the main subjects of the present paper.   

Therefore, the present paper will provide a statistical interpretation of the theory by constructing a quantum measure in Sorkin's sense \cite{So94} on its groupoid of transitions.  Such quantum measure will be determined by an invariant state of the algebra of transitions of the theory.   The key idea comes from the realization that a state $\rho$ determines a function $\varphi$ of positive semi-definite type on the given groupoid, and this function actually defines a decoherence functional in a natural way.   The relation between decoherence functionals and quantum measures allows to provide the desired statistical interpretation of the theory by identifying the amplitudes of transitions with the values $\varphi(\alpha)$ of the positive semi-definite function defining the decoherence functional, and its module square with the `probability' of an event.    In developing the theory, the GNS construction will be used to interpret the obtained notions in the more familiar terms of vector-valued measures on Hilbert spaces, and an extension of Naimark's reconstruction theorem for groupoids will be proved.  

The second part of the paper will be devoted to identify two classes of states that have a specific physical meaning.  The first one are those states such that the associated physical amplitudes satisfy Feynman's reproducing property (\ref{feynman}).  We will characterise those states in a purely algebraic way as those whose characteristic functions $\varphi$ are idempotent with respect to the convolution product in the algebra of observables of the theory.    

The second class provides an answer to the singularity of Dirac-Feynman postulate  stating that  amplitudes have the form $e^{\frac{i}{\hbar}S}$ for a quantum theory on space-time.  It will be shown that there is again a purely algebraic notion in the groupoids setting that characterises completely these states and that is called `factorization'.  The proof of the corresponding theorem is worked out in detail in the finite-dimensional case, and it constitutes one of the main results reported here. 

Finally, the third part of the paper is devoted to put Schwinger's theory of transformations functions on the same footing as the previous notions.  This is achieved by observing that there are natural states, those associated to outcomes $a$ of the theory, whose corresponding amplitudes on transitions associated to other outcomes $b$  provide precisely Schwinger's transitions functions $\langle b | a \rangle$, and they are given precisely as inner products of vectors on suitable Hilbert spaces (again obtained by a natural use of the GNS construction).

The rest of this introduction will be devoted to succinctly summarise the basic notions and notations on groupoids and their algebras used throughout the paper (see the preceding papers \cite{Ib18}, \cite{Ib18b} for a detailed account of these ideas).


\subsection{The groupoids description of Schwinger's algebra of measurements: basic notations and definitions}\label{sec:groupoids}

Even if groupoids can be described in a very abstract setting using category theory, in this paper, we will just use simple set-theoretical concepts and notations to work with them.  For the most part, we will assume that groupoids are discrete (countable) or even finite\footnote{We will be concerned mostly with the algebraic structures of the theory leaving many of the deep and delicate analytical details involved in the infinite dimensional setting for further discussion.}.  

Thus, a groupoid $\mathbf{G}$ will be a set  whose elements, denoted by greek letters $\alpha, \beta, \gamma,...$, will be called transitions. There are two maps $s,t \colon \mathbf{G} \to \Omega$, called respectively source and target, from the groupoid $\mathbf{G}$ into a set $\Omega$ whose elements will be denoted by lowercase latin letters $a,b,c,\ldots, x,y,z$ and called outcomes or events.  We will often use the diagrammatic representation $\alpha\colon a \to a'$ for the transition $\alpha$ if $s(\alpha) = a$ and $t(\alpha) = a'$.  Notice that the previous notation doesn't imply that $\alpha$ is a map from a set $a$ into another set $a'$, even if occasionally we will use the notation $\alpha(a)$ to denote $a' = t(\alpha)$. We will  also say that the transitions $\alpha$ \textit{relates} the event $a$ to the event $a'$.   

Denoting by $\mathbf{G}(y,x)$ the set of transitions relating the event $x$ with the event $y$, i.e., $\alpha \in \mathbf{G}(y,x)$ if $\alpha \colon x \to y$, there is a composition law $\circ \colon \mathbf{G}(z,y) \times \mathbf{G}(y,x) \to \mathbf{G}(z,x)$, such that if $\alpha \colon x \to y$ and $\beta \colon y \to z$, then $\beta \circ \alpha \colon x \to z$\footnote{The `backwards' notation for the composition law has been chosen so that the various representations and compositions used along the paper look more natural. It is also in agreement with the standard notation for the composition of functions.}.  Two transitions $\alpha$, $\beta$ such that $t(\alpha) = s(\beta)$ will be said to be composable.  The set of composable transitions form a subset of the Cartesian product $\mathbf{G}\times \mathbf{G}$ sometimes denoted by $\mathbf{G}_2$.   

It is postulated that the composition law $\circ$ is associative whenever the composition of three transitions makes sense, that is: $\gamma \circ (\beta \circ \alpha) = (\gamma \circ \beta) \circ \alpha$, whenever $\alpha \colon v \to x$, $\beta \colon x \to y$ and $\gamma \colon y \to z$.    Any event $a\in \Omega$ has associated a transition denoted by $1_a$ satisfying the properties $\alpha \circ 1_a = \alpha$, $1_{a'}\circ \alpha = \alpha$ for any $\alpha \colon a \to a'$.   Notice that the assignment $a \mapsto 1_a$ defines a natural inclusion $i\colon \Omega \to \mathbf{G}$ of the space of events in the groupoid $\mathbf{G}$.     Finally it will be assumed that any transition $\alpha \colon a \to a'$ has an inverse, that is there exists $\alpha^{-1} \colon a' \to a$ such that $\alpha \circ \alpha^{-1} = 1_{a'}$, and $\alpha^{-1} \circ \alpha = 1_a$.  

Given an event $x \in \Omega$, we will denote by $\mathbf{G}_+(x)$ the set of transitions starting at $a$, that is, $\mathbf{G}_+(x) = \{ \alpha \colon x \to y \} = s^{-1}(x)$.  In the same way, we define $\mathbf{G}_-(y)$ as the set of transitions ending at $y$, that is, $\mathbf{G}_-(y) = \{ \alpha \colon x \to y \} = t^{-1}(y)$.   The intersection of $\mathbf{G}_+(x)$ and $\mathbf{G}_-(x)$ consists of the set of transitions starting and ending at $x$, and is called the isotropy group $G_x$ at $x$: $G_x = \mathbf{G}_+(x) \cap \mathbf{G}_-(x)$.     

Given an event $a \in \Omega$, the orbit $\mathcal{O}_a$ of $a$ is the subset  of all events related to $a$, that is, $a' \in \mathcal{O}_a$ if there exists $\alpha \colon a \to a'$.   The isotropy groups $G_x$ and $G_y$ of two events in the same orbit, $x,y\in \mathcal{O}_a$, are isomorphic.     Clearly, the isotropy group $G_a$ acts on the right on the space of transitions starting from $a$, that is, there is a natural map $\mu_a \colon \mathbf{G}_+(a) \times G_a \to \mathbf{G}_+(a)$, given by $\mu_a(\alpha, \gamma_a) = \alpha \circ \gamma_a$ (notice that the transition $\gamma_a\colon a \to a$ doesn't change the source of $\alpha \colon a \to a'$).  Then, it is easy to check that there is a natural bijection between the space of orbits of $G_a$ in $\mathbf{G}_+(a)$ and the elements in the orbit $\mathcal{O}_a$ given by $\alpha \circ G_a \mapsto \alpha(a) = a'$. Then, we may write:
$$
\mathbf{G}_+(a) / G_a \cong \mathcal{O}_a \, .
$$
It is obvious that there is also a natural left action of $G_a$ into $\mathbf{G}_-(a)$ and that $G_a\backslash \mathbf{G}_-(a) \cong \mathcal{O}_a$ too.  We will say that the groupoid is connected or transitive if it has is a single orbit, $\Omega = \mathcal{O}_a$, for some $a$.  Then, it can be proved that $\Omega = \mathcal{O}_x$ for any $x \in \Omega$.   Any groupoid decomposes as the disjoint union of connected groupoids, any of them being the restriction of the given groupoid to any one of its orbits.   In what follows, we will always assume that groupoids are connected.

If the groupoid $\mathbf{G}$ is finite, the groupoid algebra, or algebra of virtual transitions, of the groupoid $\mathbf{G}$ is defined in the standard way as the associative algebra $\mathbb{C}[\mathbf{G}]$ generated by the elements of $\mathbf{G}$ with relations provided by the composition law of the groupoid.
That is, elements $\mathbf{a}$ in   $\mathbb{C}[\mathbf{G}]$ are finite formal linear combinations $\mathbf{a}= \sum_{\alpha \in \mathbf{G}} a_\alpha \, \alpha$, with $a_\alpha$ complex numbers. The groupoid algebra elements $\mathbf{a}$ can be thought of as generalized or mixed transitions for the system and will be called also virtual transitions.  The associative  composition law on $\mathbb{C}[\mathbf{G}]$ is defined as:
$$
\mathbf{a}\cdot \mathbf{a}' = \sum_{\alpha, \alpha' \in \mathbf{G}} a_\alpha a_{\alpha'} \, \delta_{\alpha, \alpha'} \, \, \alpha \circ \alpha'  =   \sum_{\alpha, \alpha' \in \mathbf{G}_2} a_\alpha a_{\alpha'} \, \, \alpha \circ \alpha'  \, , 
$$
where the indicator function $\delta_{\alpha, \alpha'}$ takes the value 1 if $\alpha$ and $\alpha'$ are composable, and zero otherwise.    The groupoid algebra has a natural antilinear involution operator denoted $\ast$ and defined as $\mathbf{a}^\ast = \sum_\alpha \bar{a}_\alpha\,  \alpha^{-1}$, for any $\mathbf{a}= \sum_{\alpha} a_\alpha \, \alpha$.    

If the groupoid $\mathbf{G}$ is finite, there is a natural unit element $\mathbf{1} = \sum_{a\in \Omega} 1_a$ in the algebra $\mathbb{C}[\mathbf{G}]$ (see \cite{Ib18d, Ib19} for an elementary introduction to the theory of groupoids and their representations).

Another family of relevant mixed transitions is given by the transition $\mathbf{1}_{G_a} = \sum_{\gamma_a\in G_a} \gamma_a$ which are the characteristic `functions' of the isotropy groups $G_a$, and by the transition $\mathbf{1}_{\mathbf{G}_\pm(a)} = \sum_{\alpha\in \mathbf{G}_\pm(a)} \alpha$ which represent the characteristic `functions' of the sprays $\mathbf{G}_\pm(a)$ at $a$.  Finally, we should mention the `incidence' or total transition, also called the `incidence matrix' of the groupoid, defined as $\mathbb{I} = \sum_\alpha \alpha$.



\section{Quantum measures and Schwinger's algebra}\label{sec:measure}


\subsection{Quantum measures and decoherence functionals}

Sorkin's introduction of the notion of a quantum measure allows for a statistical interpretation of Quantum Mechanics without recurring to some of the difficulties related to the existence of observers to assess the predictive capacity of the theory or the collapse of the state of the system \cite{So94}.  

According to Sorkin's theory, Quantum Mechanics can be understood as a generalized measure theory on the space $\mathcal{S}$ of all possible histories of some physical system.   It assigns a non-negative real number  $\mu(A)$, the quantum measure of $A$, to every measurable subset $A$ of the set of histories of the system.  The quantum measure $\mu$ is not an ordinary probability measure because in general the interference term:
\begin{equation}\label{I2}
I_2(A,B) = \mu (A \sqcup B) - \mu(A) - \mu(B) \, ,
\end{equation}
for two disjoint\footnote{In what follows we will use the notation $A \sqcup B$ to indicate the union of disjoint sets.} sets $A$, $B$, doesn't vanish. 
Thus, under this perspective, the feature that distinguishes a quantum theory from a classical one is interference.   This means that the quantum measure $\mu$ will enjoy different formal properties than a standard probability measure.  

We start by defining a family of set-functions, describing interference terms, for any generalized measure theory over a sample space $\mathcal{S}$ equipped with a $\sigma$-algebra $\Sigma$ of measurable sets:
\begin{eqnarray}
I_1(A) &=& \mu (A) \, , \nonumber \\ 
I_2(A,B) &=& \mu (A \sqcup B ) - \mu(A) - \mu(B) \, , \nonumber \\ 
I_3(A,B,C) &=& \mu (A \sqcup B \sqcup C ) - \mu(A \sqcup B) -  \mu(A \sqcup C) - \mu(B\sqcup C) \nonumber \\  
&&+ \mu (A ) + \mu(B) + \mu(C) \, , \label{I3}
\end{eqnarray}
and so on, where $A,B,C$ are disjoint sets of $\mathcal{S}$.   Higher order interference relations beyond 
bipartite and tripartite interference terms, as given by Eqs. (\ref{I2}), (\ref{I3}), can be defined as:
\begin{eqnarray*}
I_n (A_1, \ldots, A_n) &=& \mu (A_1 \sqcup \cdots \sqcup A_n ) - \sum_{i_1 < i_2 < \cdots < i_{n-1}}\mu(A_{i_1} \sqcup \cdots \sqcup A_{i_{n-1}}) \\ &&+ \sum_{i_1 < i_2 < \cdots <i_{n-2}}\mu(A_{i_1} \sqcup \cdots \sqcup A_{i_{n-2}}) + \cdots + (-1)^{n} \sum_{i = 1}^n \mu(A_i) \, ,
\end{eqnarray*}
for any family of disjoint sets $A_i$.
It can be shown that the interference relation $I_n$ of order $n$ implies $I_r$ for all $r \geq n$.  Actually it is easy to prove by induction that:
\begin{eqnarray}
I_{n+1} (A_0 , A_1, \ldots, A_n) &=& I_n(A_0\sqcup A_1, A_2, \ldots, A_n) \nonumber \\ &&-  I_n(A_0, A_2, \ldots, A_n) -  I_n( A_1, A_2, \ldots, A_n) \, ,  \label{recursive}
\end{eqnarray}
hence, if $I_n = 0$ on any family of disjoint measurable sets $A_i$, then $I_{n+1}$ will vanish too.

The interference functions $I_n$ allow us to distinguish between different types of theories according to their statistical properties.  According to Sorkin, a theory is of grade-$k$ if it satisfies $I_{k+1} = 0$. Thus a classical measure theory is a grade-1 measure theory, which is equivalent to saying that there is no bipartite interference, that is, $\mu (A \sqcup B ) = \mu (A) + \mu (B)$, and Kolmogorov's standard probability interpretation of the measure $\mu(A)$ can be used.   

A quantum measure theory is a grade-2 measure theory, that is, a quantum measure is a set function $\mu \colon \Sigma \to \mathbb{R}^+$ such that it satisfies the grade-2 additivity condition\footnote{Technically speaking, this definition will correspond to a pre-quantum, or finite, quantum measure, because it is necessary to add a continuity condition to make it consistent with $\Sigma$ being a $\sigma$-algebra: $\lim \mu (A_i) = \mu (\bigcap A_i)$ for all decreasing sequences, and $\lim\mu(A_i) = \mu (\bigcup A_i)$ for all increasing sequences, see \cite{Gu09} for more details.} $I_3 = 0$.   Thus a quantum measure allows to describe a theory with non-trivial second order (but no higher order) vanishing interference terms, i.e., $I_k = 0$ for  $k \geq 3$.

However, the statistical interpretation of such condition still needs to be assessed because, as it is easily exemplified by the double slit experiment\footnote{It has been shown recently though that the tripartite interference condition $I_3 = 0$ holds in quantum mechanics by using a 3-slit experiment \cite{Si09, Ne17}.}, the classical probabilistic interpretation of the measure of a set as a frequency of outcomes of a random variable cannot be held anymore.      The quantum measure of an event in Sorkin's sense is not simply the sum of the probabilities of the histories that compose it, but is given (in an extension of Born's rule) by the sum of the squares of certain sums of the complex amplitudes of the histories which comprise the event\footnote{Without entering such discussion here, Sorkin has proposed an interpretation of the number $\mu (A)$, assigned to an event $A$ by the quantum measure $\mu$, in terms of the notion of `preclusion' instead of the of notion of `expectation' \cite{So95}.  Preclusion is related to the impossibility of null sets and, in this context, it is necessary to add a regularity condition to the notion of quantum measure, that is, $\mu (A) = 0$ implies that $\mu(A \cup B) = \mu (B)$, and $\mu(A) = 0$ implies that $\mu (B) = 0$ for all $B \in \Sigma$, $B \subset A$ - in such case the quantum measure is called completely regular.}.    

More important to our interest in this paper is to understand the construction of quantum measures in the abstract background provided by the groupoid interpretation of the fundamental algebraic properties of quantum systems.     For this purpose, we will discuss first the relation of quantum measures and decoherence functionals in the realm of groupoids.  In doing so, we will extend some recent results on the representation of decoherence functionals that will be helpful in providing a new statistical interpretation of Schwinger's transformation functions.    

It should be pointed out that the recursive relation in Eq. (\ref{recursive}), when  applied to the additivity condition $I_3 = 0$, implies that $I_2$ is additive for disjoint sets $A,B,C$. More specifically, we get:
$$
I_2(A \sqcup B, C) = I_2(A , C) + I_2(B, C) \, ,
$$
In fact, if $I_2$ were additive on the first factor for all $C$ (not just for sets $C$ disjoint with $A$ and $B$), then spanning the quantity $I_2(A \sqcup B, A \sqcup B)$ we will get $\mu (A) = \frac{1}{2} I_2(A,A)$.  Then, in this case, the quantum measure $\mu$ could be recovered as a quadratic function on the algebra of measurable sets.  This suggests to consider  biadditive set functions $D \colon \Sigma \times \Sigma \to \mathbb{C}$ as a natural way of constructing quantum measures.  Actually, this idea is deeply rooted in the histories approach to quantum mechanics under the name of decoherence functionals \cite{Ge90} and, what is more important for the arguments to follow, a significant class of normalised quantum measures can be built by using decoherence functionals $D$.

Thus, a general decoherence functional on a measurable space $(\mathcal{S},\Sigma)$ is a set function $D \colon \Sigma \times \Sigma \to \mathbb{C}$ such that it is Hermitean:
\begin{equation}\label{hermitian}
D(A, B) = \overline{D(B, A)} \, , \qquad \forall A,B \in \Sigma \, ,
\end{equation}
non-negative:
\begin{equation}\label{positivity}
D(A,A) \geq 0 \, ,
\end{equation}
and addtive:
\begin{equation}\label{finite_additivity}
D(A\sqcup B, C) = D(A, C) + D(B, C) \, , \quad \forall A,B,C \in \Sigma \, , \quad A\cap B = \emptyset \, .
\end{equation}

It will be assumed that the decoherence functional\footnote{The notion of decoherence functional is known under the name of bimeasures in abstract measure theory and has been discussed thoroughly in multiple contexts, see for instance \cite{Ib14} for the description of the moment problem for polymeasures and references therein.} $D$ is normalised, that is, $D(\mathcal{S},\mathcal{S}) = 1$, and, consequently, this notion is sufficient to construct a quantum measure by means of:
\begin{equation}\label{muD}
\mu (A) = D(A, A) \, .
\end{equation}

However, in order to obtain a continuous completely regular quantum measure (just a quantum measure for short in what follows) it is necessary to introduce a slightly more restrictive definition of decoherence functional \cite{Do10}.  A strongly positive normalised decoherence functional $D$ is a complex-valued set function defined on the Cartesian product $\Sigma \times \Sigma$ satisfying the following properties:

\begin{enumerate}

\item[i.-] Normalization: 
\begin{equation}\label{decoherence_normalization}
D(\mathcal{S},\mathcal{S}) = 1 \, .
\end{equation}

\item[ii.-] $\sigma$-Additivity:  $D(\cdot , A)$ is a complex measure for any $A \in \Sigma$.

\item[iii.-] Positivity:  Given any natural number $n$ and any family $A_1, \ldots, A_n$ of measurable sets  in $\Sigma$, then $D(A_i,A_j)$ is a positive semi-definite $n\times n$-matrix.

\end{enumerate}

Condition (i) is an irrelevant normalisation condition.  Notice that condition (iii) implies conditions (\ref{hermitian}) and (\ref{positivity}) above, while condition (ii) implies the finite-additivity condition (\ref{finite_additivity}).   Then, it is a routine check to show  that the set function $\mu$ defined by Eq. (\ref{muD}) is a completely regular quantum measure (see for instance \cite{Gu09}).


\subsection{Quantum measures on groupoids}
 
As it turns out, the groupoid formalism to describe quantum systems provides a natural framework to construct decoherence functionals, hence, to build quantum measures, and thus it provides a statistical interpretation of the theory.   

We will consider that the connected discrete\footnote{As customary in this series of papers we will assume that the groupoid $\mathbf{G}$ is discrete countable (or even finite) to avoid the technical complications brought by functional analysis, even though most of the theory can be extended naturally to continuous or Lie groupoids with ease as it will be shown elsewhere.} groupoid $\mathbf{G} \rightrightarrows\Omega$ provides a description of our quantum system.  


\subsubsection{Decoherence functionals and positive semidefinite functions on groupoids}

Because of its $\sigma$-additivity, a decoherence functional $D$ on the discrete groupoid $\mathbf{G}$ is determined by their values on singletons\footnote{The $\sigma$-algebra of measurable sets is just the power set of $\mathbf{G}$, that is $\Sigma = \mathcal{P}(\mathbf{G})$.}, that is, 
$$
D(A,B) = \sum_{\alpha\in A, \beta\in B} D(\{\alpha \}, \{\beta \}) \, ,
$$ 
then, we may consider a decoherence functional on discrete groupoids as defined by a bivariate function $\Phi \colon \mathbf{G} \times \mathbf{G} \to \mathbb{C}$, $ \Phi (\alpha, \beta ) = D(\{\alpha \}, \{\beta \})$,  satisfying:
\begin{enumerate}
\item \begin{equation}\label{normalization_bivariate} 
\sum_{\alpha, \beta \in \mathbf{G}} \Phi (\alpha,\beta) = 1 \, ;
\end{equation}

\item Given any natural number $n$ and any family $\alpha_1, \ldots, \alpha_n$ of transitions in $\mathbf{G}$, then $\Phi(\alpha_i,\alpha_j)$ is a positive semi-definite $n\times n$-matrix.
\end{enumerate}
The first property (1) is the immediate consequence of condition (i) in the definition of decoherence functionals, Eq. (\ref{decoherence_normalization}), and the second condition is equivalent to condition (iii) (notice that the sum of positive semi-definite matrices are positive semi-definite).  
Consistently, we will also say that a bivariate function $\Phi$ satisfying condition (2) above is positive semi-definite.   

Let us introduce another notion which is  relevant for the purposes of this paper.  A function $\varphi \colon \mathbf{G} \to \mathbb{C}$ will be said to be positive semi-definite if for any $n \in \mathbb{N}$, $\xi_i \in \mathbb{C}$, $\alpha_i \in \mathbf{G}$, $i = 1, \ldots, n$, the following inequality is satisfied:
$$
\sum_{i,j= 1}^n \bar{\xi}_i \xi_j \, \varphi (\alpha_i^{-1}\circ \alpha_j ) \geq 0 \, ,
$$
where the sum is taken over all pairs $\alpha_i, \alpha_j$ such that the composition $\alpha_i^{-1}\circ \alpha_j$ makes sense, that is, $t(\alpha_j) = t(\alpha_i)$.   If we want to emphasise that the sum is restricted to those pairs $\alpha_i$ and $\alpha_j$ such that  $\alpha_i^{-1}$ and $\alpha_j$ are composable we will also write:
$$
\sum_{i,j= 1}^n \bar{\xi}_i \xi_j \, \varphi (\alpha_i^{-1}\circ \alpha_j ) \, \delta (t(\alpha_i), t(\alpha_j)) \geq 0 \, ,
$$
where the delta function $\delta (t(\alpha_i), t(\alpha_j))$ implements the composability condition above.  Note that $\varphi (1_x)$ must be a non-negative real number and that, if $\lambda \geq 0$, then $\lambda \varphi$ is positive semi-definite for any positive semi-definite function $\varphi$. We will say that $\varphi$ is normalized if $\sum_{x\in \Omega} \varphi (1_x) = 1$, and we will always assume  this to be the case in the following.

Clearly, any positive semi-definite function $\varphi$ on $\mathbf{G}$ defines a bivariate positive semi-definite function $\Phi$ by means of:
\begin{equation}\label{phiPhi}
\Phi (\alpha, \beta ) = \delta (t(\alpha), t(\beta)) \, \varphi( \alpha^{-1}\circ \beta ) \, , \qquad \alpha, \beta \in \mathbf{G} \, .
\end{equation}

Among the decoherence functionals $D$ on groupoids, the invariant ones play a distinguished role.  A decoherence functional $D$ on the groupoid $\mathbf{G}$ is said to be (left-) invariant if $D(\alpha\circ A, \alpha\circ B) = D(A, B)$ for all subsets $A,B$.  

In terms of the corresponding bivariate positive semidefinite function $\Phi$, a decoherence functional is invariant iff $\Phi$ is (left-) invariant with respect to the natural action of the groupoid $\mathbf{G}$ on the product groupoid $\mathbf{G}\times \mathbf{G}$, that is\footnote{A similar definition can be used for right-invariant decoherence functionals.},: 
\begin{equation}
\Phi (\alpha\circ \beta, \alpha\circ \beta') = \Phi(\beta, \beta') \, ,
\end{equation}
for all triples $\alpha, \beta, \beta'$ such that the compositions $\alpha\circ \beta$ and $\alpha\circ \beta'$ make sense.
Then, it is clear that there is a one-to-one correspondence between invariant strongly positive decoherence functionals $D$ and positive semidefinite functions $\varphi$ on the groupoid $\mathbf{G}$, namely, the correspondence given by the assignment $\varphi \mapsto \Phi$ given in Eq. (\ref{phiPhi}).  

Notice that the converse of (\ref{phiPhi}) is given by: $\Phi \mapsto \varphi$, with $\varphi(\alpha) = \Phi(1_y,\alpha ) = \Phi (\alpha^{-1}, 1_x)$, if $\alpha \colon x \to y$.  

The previous discussion shows that we may study invariant decoherence functionals on discrete groupoids, that is, invariant quantum measures on them, by studying the corresponding positive semi-definite functions $\varphi$.  On the other hand, a natural way to study decoherence functionals (and almost any abstract object in mathematics) is by looking at their representations (see for instance \cite{Gu12} where recent results on this direction are shown).    The relevant observation here is that positive semi-definite functions on groupoids provide a natural way to construct representations of groupoids and simultaneously of decoherence functionals.   Such theory extends naturally that of positive semi-definite functions on groups with an analogue of Naimark's reconstruction theorem for groups that provides a natural representation for the decoherence functional associated to the function $\varphi$.   We will devote the following paragraphs to develop the theory in the case of discrete groupoids we are working with.


\subsubsection{States and positive semidefinite functions on groupoids}\label{sec:states_decoherence}

Given the groupoid $\mathbf{G}\rightrightarrows \Omega$, its associated $C^*$-algebra $C^*(\mathbf{G})$ provides the background for the amplitudes and for the algebra of observables of the theory \cite{Ib18b}.   

There are various ways of constructing a $C^*$-algebra associated to the groupoid $\mathbf{G}\rightrightarrows \Omega$.  In the finite case, it can be identified with the algebra $\mathcal{F}(\mathbb{G})$ of functions on $\mathbf{G}$, or with the algebra $\mathbb{C}[\mathbf{G}]$ of finite linear combinations of elements in $\mathbf{G}$, recall Sect. \ref{sec:groupoids}.   In the countable discrete case, we may consider the von Neuman algebra generated by the family of operators on $L^2(\mathbf{G})$ defined by the regular representation, but, in our context, and  provided that the fundamental representation $\pi_0$ of the groupoid is faithful, we will simplify the discussion by considering the closure of $\mathbb{C}[\mathbf{G}]$ with respect the norm induced from its fundamental representation (see for instance the construction of the $C^*$-algebra of the groupoid $\mathbf{A}_\infty$ in \cite{Ib18b}).   In other words, consider the Hilbert space $L^2(\Omega)$.
If $\Omega$ is countable, $L^2(\Omega)$ is just the complex separable Hilbert space generated by $\Omega$ with the inner product defined by declaring that the elements $x$ of $\Omega$ form an orthonormal basis $\{|x\rangle \}$.   The fundamental representation $\pi_0 \colon \mathbb{C}[\mathbf{G}] \to \mathcal{B} (L^2(\Omega))$ is given by:
\begin{equation}\label{fundamental}
(\pi_0(\mathbf{a}) \psi) (x) = \sum_{\alpha \in \mathbf{G}_+(x)} a_\alpha \psi (t(\alpha)) \, .
\end{equation}
Notice that $\pi_0(\mathbf{a}^*) = \pi_0(\mathbf{a})^\dagger$.   Then, we may consider the von Neumann algebra generated by the operators $\pi_0(\mathbf{a})$, that is, $C^*(\mathbf{G}) = (\pi_0(\mathbb{C}[\mathbf{G}] ))'' \subset  \mathcal{B} (L^2(\Omega))$.   It is also clear that if the fundamental representation is faithful\footnote{Which requires that $\Omega$ is large enough, for instance the fundamental representation of a group is not faithful as $\Omega$ consists of just one element.}, then $\mathbf{G}$ is mapped injectively in $ \mathcal{B} (L^2(\Omega))$ and, as is easily checked, the algebra  $ C^*(\mathbf{G}) $, that in what follows  will be denoted also  as $\mathcal{A}_{\mathbf{G}}$, is unital with unit the identity operator $\mathbf{1} = I$.

Given a unital $C^*$-algebra, a state $\rho$ on it is a normalised positive linear functional.
In the previous situation, a state will be a linear map $\rho \colon \mathcal{A}_{\mathbf{G}} \to \mathbb{C}$ such that $\rho (\mathbf{1}) = 1$ and $\rho (\mathbf{a}^* \cdot \mathbf{a}) \geq 0$ for all $\mathbf{a}$.   States play a particularly relevant role in the study of $C^*$-algebras.  The space of states form a convex domain in the dual space of the algebra denoted as $\mathcal{S}(\mathcal{A}_{\mathbf{G}})$ (or just $\mathcal{S}$ for short) and it is well-known that the structure of the algebra can be recovered from them.    

In the discussion to follow, states are going to play an instrumental role because of the GNS construction and of the following observation:  there is a one-to-one correspondence between states and continuous positive semi-definite functions $\varphi$ on $\mathbf{G}$.  The correspondence is as follows.  Let $\varphi \colon \mathbf{G} \to \mathbb{C}$ be a normalized positive semi-definite function, then, we define the linear map $\rho_{\varphi} \colon \mathcal{A}_\mathbf{G} \to \mathbb{C}$ as (we consider for simplicity that $\Omega$ is finite\footnote{In the continuous or infinite case, it will be assumed that $\Omega$ carries a probability measure $\nu$, the one used to define $L^2 (\Omega, \nu)$ and $|\Omega| = 1$.}):
$$
\rho_\varphi (\mathbf{a}) = \sum_\alpha a_\alpha\, \varphi (\alpha ) \, .
$$
In the finite case, $\mathbf{1} = \sum_{x\in \Omega} 1_x$, and, clearly, $\rho_\varphi (\mathbf{1}) = 1$.  Moreover, a simple computation shows that:
\begin{equation}\label{rhophi}
\rho_\varphi (\mathbf{a}^* \cdot \mathbf{a}) = \sum_{(\alpha^{-1}, \beta) \in \mathbf{G}_2} \bar{a}_\alpha a_\beta  \, \varphi(\alpha^{-1}\circ \beta) \geq 0 \, ,
\end{equation}
by the very definition of $\varphi$.  Conversely, given a state $\rho$ on $\mathcal{A}_\mathbf{G}$, we define the function $\varphi$ on $\mathbf{G}$ by restriction of $\rho$, that is, we set:
$$
\varphi_\rho (\alpha ) = \rho (\alpha ) \, ,
$$
and, clearly, $\varphi_\rho$ is normalized positive semi-definite because Eq. (\ref{rhophi}) can be read backwards.  In this case, we will say that $\varphi_\rho$ is the characteristic function of the state $\rho$.

We conclude this section by realising that states on the algebra of generalised transitions of the system are associated with positive semidefinite functions on the groupoid, hence, they determine invariant decoherence functionals, and, consequently,  invariant quantum measures on $\mathbf{G}$.
Therefore, in particular, in the case of finite groupoids, there is a one-to-one correspondence between states and invariant quantum measures on the groupoid.   Notice that in this case, if $A \subset \mathbf{G}$, we get:
\begin{eqnarray}
\mu_\rho (A) &=& D(A,A) = \sum_{\alpha,\beta \in A} \Phi(\alpha, \beta) = \sum_{\alpha,\beta \in A} \delta(t(\alpha), t(\beta) )\, \varphi(\alpha^{-1} \circ \beta) \nonumber \\ &=& \sum_{\alpha,\beta \in A}  \delta(t(\alpha), t(\beta) ) \, \rho (\alpha^* \cdot \beta) \label{quantum_rho} \, .
\end{eqnarray}
The remarkable formula (\ref{quantum_rho}) embodies, in the abstract groupoid formalism ,Sorkin's quantum measure expression for systems described on spaces of histories\footnote{It is also remarkable that the delta function can be dropped in the last expression from (\ref{quantum_rho}) because if $\alpha^{-1}$ and $\beta$ are not composable, then $\alpha^* \cdot \beta = 0$.} (see for instance \cite[eq. 14]{So16}) and explains the quadratic dependence of quantum measures on physical transitions.   

Notice that in the context developed in this section, the evaluation of the state $\rho$ on a transition $\alpha$ can be thought as the complex amplitude of the physical transition defined by $\alpha$, thus, the previous formula encodes the rule that `probabilities' are obtained by module square of amplitudes in the abstract setting of groupoids.   The previous expression for the quantum measure (and the decoherence functional) is given in abstract terms and we would like to describe them in terms of a concrete realization of the theory on a Hilbert space.  This will be the subject of the following sections.



\section{Representations of decoherence functionals and quantum measures}\label{sec:quantum_measure}


\subsection{Representations of groupoids and algebras}

The background needed to construct representations of decoherence functionals on Hilbert spaces in the groupoid formalism of quantum mechanics will be provided by the representations of the groupoid $\mathbf{G} \rightrightarrows \Omega$ itself.  
Even if a functorial definition of representations of groupoids could be used (see the recent presentation of the basic theory \cite{Ib19}), in the setting described in the previous sections, it is simpler to define a representation of the groupoid $\mathbf{G}$ as a representation of the $C^*$-algebra $\mathcal{A}_\mathbf{G}$ on the $C^*$-algebra $\mathcal{B}(\mathcal{H})$ of bounded operators on a complex separable Hilbert space $\mathcal{H}$, that is, we consider a $C^*$-algebra homomorphism $\pi \colon \mathcal{A}_\mathbf{G} \to \mathcal{B}(\mathcal{H})$ which is continuous in the sense that for any $\psi \in \mathcal{H}$, the map $\mathbf{a} \to || \pi(\mathbf{a})\psi ||$ is continuous.  Notice that $\pi (\mathbf{1}) = I$ and $\pi(\mathbf{a}^*) = \pi (\mathbf{a})^\dagger$.  In particular, the fundamental representation $\pi_0$ discussed before, Eq. (\ref{fundamental}), is an example of an irreducible representation of the groupoid $\mathbf{G}$.   

The theory of representations of groupoids shares many aspects with the theory of representations of groups (at least in the finite case, this relation is well developed, see, for instance, \cite{Ib19}), but we will not pretend to start such general discussion here.   In what follows, we will just depart from a given state to construct explicit representations of the groupoid by means of the so called GNS construction.  

Before describing this idea, we would like to point out that if $\pi$ is a nondegenerate representation\footnote{That is, a representation such that $\overline{\mathrm{span} \{ \pi(\mathbf{a})\psi \mid \mathbf{a} \in \mathcal{A}_\mathbf{G}, \psi \in \mathcal{H}\}}  = \mathcal{H}$.} of the groupoid algebra $\mathcal{A}_\mathbf{G}$ on the Hilbert space $\mathcal{H}$, and $\psi$ is a cyclic vector for such representation\footnote{That is, the family of vectors $\{\pi(\mathbf{a}) \psi\}_{\mathbf{a}\in\mathcal{A}_{\mathbf{G}}}$ span $\mathcal{H}$.}, then, we may define the positive semi-definite function:
\begin{equation}\label{character}
\varphi_{\pi,\psi} (\alpha) = \langle \psi , \pi(\alpha) \psi \rangle \, ,
\end{equation}
associated to the representation $\pi$ and the cyclic vector $\psi$.

Notice that (\ref{character}) actually defines a positive semi-definite function on $\mathbf{G}$ as it is shown by the following simple computation (as usual, the sums are taken over all composable pairs $\alpha_i^{-1}$, $\alpha_j$): 
\begin{eqnarray*}
\sum_{i,j=1}^n \bar{\xi}_i \xi_j \varphi_{\pi,\psi} (\alpha_i^{-1}\circ \alpha_j) &=& \sum_{i,j=1}^n \bar{\xi}_i \xi_j \langle \psi ,\pi(\alpha_i^{-1}\circ \alpha_j) \psi \rangle  =  \sum_{i,j=1}^n \bar{\xi}_i \xi_j \langle \psi ,\pi(\alpha_i)^\dagger\pi( \alpha_j) \psi \rangle  \\ &=&  \sum_{i,j=1}^n \bar{\xi}_i \xi_j \langle \pi(\alpha_i) \psi ,\pi( \alpha_j) \psi \rangle \leq \langle  \sum_{i=1}^n \xi_i \pi(\alpha_i) \psi , \sum_{j=1}^n \xi_j \pi( \alpha_j) \psi \rangle  \\ &=&  || \sum_{j=1}^n \xi_j \pi( \alpha_j) \psi ||^2 \geq 0 \, .
\end{eqnarray*}
Then, if $\psi$ is normalized, the state defined by $\varphi_{\pi,\psi}$ determines a quantum measure $\mu_{\pi,\psi}$ given by:
\begin{eqnarray}
\mu_{\pi,\psi} (A) &=& D_{\pi,\psi}(A,A) = \frac{1}{Z_0 }\sum_{\alpha,\beta \in A} \delta(t(\alpha), t(\beta) )\varphi_{\pi,\psi}(\alpha^{-1} \circ \beta) \nonumber \\ &=& \frac{1}{Z_0 } \sum_{\alpha,\beta \in A}  \delta(t(\alpha), t(\beta) ) \, \langle \pi(\alpha) \psi ,\pi( \beta) \psi \rangle  \label{quantum_pi_psi} \, ,
\end{eqnarray}
where $Z_0$ is an appropriate normalization factor determined by the normalization condition (1) in (\ref{decoherence_normalization}).
In other words, we may define (up to a normalization constant) a vector-valued measure $\nu_\pi \colon \Sigma \to \mathcal{H}$ given by:
$$
\nu_\pi (A) = \sum_{\alpha \in A} \pi (\alpha) \psi \, ,
$$
that represents the decoherence functional $D_{\pi,\psi}$ associated to the quantum measure $\mu_{\pi,\psi}$ (see \cite{Do10b}, \cite{Gu12} for an account of the general theory).  Notice, finally, that the cyclic vector $\psi$ for the representation $\pi$ defines a state $\rho_{\pi,\psi} (\mathbf{a}) = \sum_\alpha a_\alpha \langle \psi, \pi (\alpha) \psi \rangle$ whose associated quantum measure is exactly the one defined in Eq. (\ref{quantum_pi_psi}).

It should also pointed out that the characteristic function $\varphi_{\pi, \psi}$ can  be expressed as:
$$
\varphi_{\pi, \psi} (\alpha) = \mathrm{Tr\,} (\hat{\rho}_\psi \pi (\alpha)) \, , 
$$
where $\hat{\rho}_\psi$ denotes the rank-one orthogonal projector $|\psi\rangle \langle \psi| $ on $\mathcal{H}$ onto the one-dimensional space spanned by the vector $|\psi\rangle$.   Then, if we consider instead the trivial projector defined by the identity operator $I$, we will get:
$$
\varphi_{\pi, I} (\alpha) = \mathrm{Tr\,} (\pi (\alpha)) = \chi (\alpha) \, , 
$$
where the function $\chi = \varphi_{\pi, I} $ is commonly known as the character of the representation $\pi$.  It is because of this instance that we would like to call the positive semi-definite function $\varphi_{\pi, \psi}$  the \textit{smeared} character of the representation $\pi$ with respect to the state $\psi$. 


\subsection{The GNS construction.  Representations associated to states}\label{sec:GNS}

Because a quantum measure  $\mu$, or for that matter, a decoherence functional, is associated to a state  $\rho$ on the algebra of the system, it is just natural to use $\rho$ to build a specific representation of the algebra itself.   The GNS construction is the well-known procedure to build a representation of the algebra given a state on it, and we will succinctly review it in the present context.

Consider a state $\rho$ on $\mathcal{A}_{\mathbf{G}}$.   There is a canonical Hilbert space $\mathcal{H}_\rho$ associated to it defined as the completion of the quotient linear space $\mathcal{A}_{\mathbf{G}}/\mathcal{J}_\rho$, where $\mathcal{J}_\rho = \{ \mathbf{a} \mid \rho (\mathbf{a}^\ast \cdot \mathbf{a} ) = 0\}$ denotes the Gelfand ideal of $\rho$, with respect to the norm $||\cdot ||_\rho$ associated to the state $\rho$ and defined by:
$$
||\, [\mathbf{a}]\, ||_\rho = \rho (\mathbf{a}^\ast \cdot \mathbf{a} ) \, ,  
$$
where $[\mathbf{a}] =  \mathbf{a} + \mathcal{J}_\rho$ is in $\mathcal{A}_{\mathbf{G}}/\mathcal{J}_\rho $.
Thus, the Hilbert space $\mathcal{H}_\rho = \overline{\mathcal{A}_{\mathbf{G}}/\mathcal{J}}_\rho^{||\cdot ||_\rho}$ will be called the GNS Hilbert space associated to the state $\rho$\footnote{Such Hilbert space has been recognized in a closely related context by Dowker and Sorkin on its histories interpretation of quantum measures under the name of the `history Hilbert space' \cite{Do10}.}.  The parallelogram identity implies that the inner product $\langle \cdot, \cdot \rangle_\rho$ on $\mathcal{H}_\rho$ is given by:
\begin{equation}\label{inner_rho}
\langle [\mathbf{a}] , [\mathbf{b}] \rangle_\rho = \rho (\mathbf{a}^\ast \cdot \mathbf{b} ) \, .
\end{equation}

For our purposes it is fundamental to observe that there is a natural representation $\pi_\rho$ of the $C^*$-algebra $\mathcal{A}_\mathbf{G}$ on $\mathcal{H}_\rho$ defined by:
$$
\pi_\rho (\mathbf{a}) ([\mathbf{b}]) = [\mathbf{a}\cdot \mathbf{b}] \, , 
$$
for all $\mathbf{a}\in \mathcal{A}_\mathbf{G}$ and $[\mathbf{b}] \in \mathcal{H}_\rho$.  Clearly, the unit $\mathbf{1}$ of the algebra $\mathcal{A}_\mathbf{G}$ is mapped into the identity operator $I$ and $\pi_\rho(\mathbf{a}^\ast ) = \pi_\rho (\mathbf{a})^\dagger $.   

The representation $\pi_\rho$ is non-degenerate and the unit element $\mathbf{1}$ provides a cyclic vector for it.  Denoting, as customary, by $|0\rangle$ the vector $[\mathbf{1}] \in \mathcal{H}_\rho$, it is clear that, the subspace of vectors of the type $\pi_\rho(\mathbf{a}) |0\rangle$ with $\mathbf{a}\in\mathcal{A}_{\mathbf{G}}$, is dense in  $\mathcal{H}_\rho$.  The vector $|0\rangle$ is called (context depending) the \textit{ground}, \textit{vacuum} or \textit{fundamental} vector of the GNS Hilbert space $\mathcal{H}_\varphi$, and we have  $\langle 0 \mid 0 \rangle = \rho (\mathbf{1}^\ast \cdot \mathbf{1}) = 1$.


\subsection{Representation of decoherence functionals}

We shall consider now the state $\rho$ associated to a given invariant decoherence functional $D$.  In other words, according to the discussion in Sect. \ref{sec:states_decoherence}, we may consider a continuous positive semi-definite function $\varphi$ on the groupoid $\mathbf{G}$ and the state $\rho_\varphi$ (and the corresponding decoherence functional) associated to it (recall the fundamental equation relating all these notions, Eq. (\ref{quantum_rho})).  Denoting the GNS Hilbert space associated to the state $\rho_\varphi$ by $\mathcal{H}_\varphi$, we get that $\mathcal{H}_\varphi$ is the completion of $\mathcal{A}_\mathbf{G}/\mathcal{J}_\varphi$, where $\mathcal{J}_\varphi$ denotes now the Gelfand's ideal:
$$
\mathcal{J}_\varphi = \{ \mathbf{a} \mid \sum_{t(\alpha) = t(\beta)} \bar{a}_\alpha a_\beta \,  \varphi (\alpha^{-1}\circ \beta) = 0 \} \, ,
$$
with respect to the norm:
$$
|| \,[\mathbf{a}]\, ||^2_\varphi =   \sum_{t(\alpha) = t(\beta)} \bar{a}_\alpha a_\beta \,  \varphi (\alpha^{-1}\circ \beta) \, ,
$$
that defines the inner product in $\mathcal{H}_\varphi$:
\begin{equation}\label{inner_phi}
\langle [\mathbf{a}] , [\mathbf{b}] \rangle_\varphi = \rho_\varphi (\mathbf{a}^\ast \cdot \mathbf{b} ) =  \sum_{t(\alpha) = t(\beta)} \bar{a}_\alpha b_\beta \, \varphi (\alpha^{-1}\circ \beta) \, ,
\end{equation}
with $\mathbf{a} = \sum_\alpha a_\alpha \, \alpha$,  $\mathbf{b} = \sum_{\beta} b_\beta \, \beta$.
The GNS representation $\pi_\varphi$ defined by the state $\rho_\varphi$ and the fundamental vector $|0\rangle$ allows us to write the amplitude $\varphi (\mathbf{a})$ in the suggestive way:
\begin{equation}\label{amplitude_phi}
\varphi (\mathbf{a}) = \rho_\varphi (\mathbf{a}) = \rho_\varphi (\mathbf{1}^\ast \cdot \mathbf{a}) = \langle 0 \mid [\mathbf{a}] \rangle_\varphi = \langle 0 \mid \pi_\varphi(\mathbf{a}) \mid 0\rangle_\varphi  \, ,
\end{equation}
where we have used (\ref{inner_phi}) and the canonical representation $\pi_\varphi (\mathbf{a}) |0\rangle = [\mathbf{a}]$.

In the same spirit as Eq. (\ref{amplitude_phi}), the canonical representation of the algebra of transitions $\mathcal{A}_\mathbf{G}$ provided by the positive semi-definite function $\varphi$, allows to provide a representation of the decoherence functional in terms of amplitudes in the Hilbert space $\mathcal{H}_\varphi$ (and it constitutes also the particular instance of Eq. (\ref{quantum_pi_psi})) given by:
$$
D_\varphi(\alpha, \beta) = \varphi(\alpha^{-1}\cdot \beta) = \langle 0 \mid \pi_\varphi(\alpha)^\dagger \pi_\varphi(\beta) \mid 0 \rangle_\varphi  \, .
$$
Notice that, if $t(\alpha) \neq t(\beta)$, then $\alpha^{-1}$ and $\beta$ are not composable and $\alpha^{-1}\cdot \beta = 0$.
Hence, $\langle [\alpha] \mid [\beta] \rangle_\varphi = 0$, or, equivalently, $D_\varphi (\alpha, \beta) = 0$.
In this case, we will also say,  mimicking the histories based approach to quantum mechanics, that the two transitions are decoherent. 

Finally, notice that, on singletons, the quantum measure $\mu_\varphi$ determined by the state $\rho_\varphi$ has the definite expression:
$$
\mu_\varphi (\{\alpha \}) = D_\varphi (\alpha, \alpha) = || \pi_\varphi(\alpha) |0\rangle ||_\varphi^2 \, ,
$$
and this expression presents $\mu_\varphi$ as the module square of an amplitude.
However, the non-additivity of the quantum measure implies that, for subsets that are not singletons, the computation of $\mu_\varphi$ has to be performed according to the superposition rule provided by Eq. (\ref{quantum_pi_psi}).


\subsection{Naimark's reconstruction theorem for groupoids}

The discussion in the previous section can be summarised in the form of a theorem:

\begin{theorem}\label{reconstruction} Let $\mathbf{G}\rightrightarrows \Omega$ be a discrete groupoid with finite space $\Omega$.
Then, for any positive semi-definite function $\varphi$ on $\mathbf{G}$, there exists a Hilbert space $\mathcal{H}$, a unitary representation $\pi$ of the groupoid $\mathbf{G}$ on $\mathcal{H}$, and a vector $|0\rangle$ such that:
$$
\varphi (\alpha ) = \langle 0 | \pi (\alpha) |0 \rangle \, .
$$
In other words, any positive semi-definite function $\varphi$ on a groupoid is the smeared character of a representation of the groupoid.
\end{theorem}

This statement can be considered as the extension of Naimark's reconstruction theorem for groupoids (admittedly, the particular instance of discrete groupoids with finite space of events).   The `reconstruction' character of the previous theorem is justified from the following considerations.    

Let $\pi$ be a unitary representation of the groupoid $\mathbf{G}$  on the Hilbert space $\mathcal{H}$ (by that we mean that $\pi$ defines a $C^*$ representation of the $C^*$-algebra of the groupoid $\mathbf{G}$ on the $C^*$-algebra of bounded operators on the Hilbert space $\mathcal{H}$).   Consider now a state $\rho$ of the $C^*$-algebra $\mathcal{B}(\mathcal{H})$. Because of Gleason's theorem such state can be identified with a normalised Hermitean nonnegative operator $\hat{\rho}$. Then we define the function:
\begin{equation}\label{H_pi_rho}
\varphi_\rho (\alpha) = \mathrm{Tr \,}(\hat{\rho}\,  \pi (\alpha)) \, .
\end{equation}
It is immediate to check that $\varphi_\rho$ defines a positive semidefinite function on $\mathbf{G}$.  Then Thm. \ref{reconstruction} shows that there exists a Hilbert space $\mathcal{H}'$, a representation $\pi'$ and a state $\rho'  = | 0\rangle \langle 0 |$, such that:
\begin{equation}\label{Hprime}
\varphi_\rho (\alpha ) = \mathrm{Tr\,} (\rho' \, \pi'(\alpha)) \, .
\end{equation}
However,  we must point out that, in principle, both representations of the function $\varphi$ provided by Eqs. (\ref{H_pi_rho}) and (\ref{Hprime}) are not equivalent.  In the particular instance of groups, there is a positive answer to the previous question when the representation $\pi$ is irreducible.  In the more general situation of groupoids, these issues will be properly discussed elsewhere.



\section{Factorizing states and decoherence functionals}

The general discussion of Sect. \ref{sec:quantum_measure} has provided a general framework for a statistical interpretation of a groupoids based quantum theory by the hand of quantum measures and their realization by means of states on the algebra of amplitudes of the theory.
However, no specific properties of the states have been identified that will reflect relevant physical properties of the system.   

In this section, we will discuss first the class of states (or quantum measures) the elements of which satisfy Feynman's composition of amplitudes law (\ref{feynman}), and we will identify a particular family of states, that will be called factorizing states, strongly suggesting a Lagrangian based sum-over-histories interpretation of the corresponding quantum measure.  We will close in this way the loop started by Dirac's insight on the role played by the Lagrangian in quantum mechanics and the answers provided by Feynman and Schwinger to that question as discussed in the introduction.  

In the remaining of this section, as stated already before and in order to simplify the presentation, we will restrict ourselves to the case of finite groupoids (even if the formalism extends naturally to countable discrete or even continuous groupoids).


\subsection{Reproducing states}\label{sec:reproducing}

States are just normalised positive linear functionals on the $C^*$-algebra of the groupoid, hence, they are blind to the specific details of the algebraic structure of the algebra (they just preserve the positive cone of the algebra).   It is true though that the $C^*$-algebra structure can be recovered from the space of states, more precisely, because of Kadison's theorem \cite{Ka51}, the real part of a $C^*$-algebra is isometrically isomorphic to the space of all $w^*$-continuous affine functions on its state space, and then, as it was shown by Falceto \textit{et al}, the $C^*$-algebra can be constructed on the space of affine function on the state space iff such space has the structure of a Lie-Jordan-Banach algebra \cite{Fa13} (see also \cite{AlfShu,AlfShu2}).    

Thus, in general, the amplitudes $\varphi(\alpha)$ associated to a given state (or quantum measure) do not satisfy any additional property related to the structure of the algebra.
In particular, they do not satisfy the reproducing property characteristic of Feynman's sum-over-histories interpretation of quantum mechanics discussed in the previous section.  However, it is not hard to characterise a class of states such that the reproducing formula  given by Eq. (\ref{feynman}), that can also be called the abstract Chapman-Kolmogorov equation, holds.

The reproducing condition states will be characterized in terms of the corresponding positive semi-definite function $\varphi$ associated to them.   Because $\varphi \colon \mathbf{G} \to \mathbb{C}$ is a function defined on the groupoid, it is convenient to describe first the structure of the algebra $\mathcal{F}(\mathbf{G})$ of functions on the groupoid.     In the case of finite groupoids, such algebra can be identified with the algebra of amplitudes (see \cite{Ib18b}).     In any case, the associative product $\star$ in $\mathcal{F}(\mathbf{G})$,  called the convolution product, is the natural one induced from the groupoid composition law and is defined by the standard formula:
$$
(f\star g)(\gamma) = \sum_{\tiny{\begin{array}{c}(\alpha, \beta) \in \mathbf{G}_2 \\ \alpha \circ \beta = \gamma \end{array}}} f(\alpha) g (\beta) \, , \qquad f,g \in \mathcal{F}(\mathbf{G}) \, , \gamma \in \mathbf{G} \, .
$$

As in the case of the algebra $\mathcal{A}_\mathbf{G}$, if the space of events $\Omega$ is finite, there is a natural unit element, denoted again by $\mathbf{1}$, and defined as $\sum_{x\in \Omega} \delta_x$, with $\delta_x$ the function that takes the value 1 at $1_x$ and zero otherwise.   

In addition to the associative structure, there is also an
antiunitary involution operator $(\cdot)^\ast$ given by $f^\ast(\alpha) = \overline{f(\alpha^{-1})}$.    The $\ast$-algebra $\mathcal{F}(\mathbf{G})$, like the algebra $\mathbb{C}[\mathbf{G}]$, has a natural representations on the space of square integrable functions on $\Omega$ and $\mathbf{G}$ itself, denoted with the symbols $\pi_0$, $\pi_R$ and $\pi_L$, and called, respectively, the fundamental, right and left regular representations.   The regular representation allows to define the von Neumann algebra of the groupoid as the weak closure of the range $\pi_R(\mathbf{G})$ in the algebra $\mathcal{B}(L^2(\mathbf{G}))$, provided that a suitable measure has been chosen in the groupoid\footnote{Contrary to the situation with groups, even if $\mathbf{G}$ is locally compact there is not a canonical (right/left) `invariant' measure on the groupoid, but a family of Haar measures had to be chosen, see \cite{Re80, La98} for details.}.

We will say that a positive semi-definite function $\varphi$ has the reproducing property if it satisfies 
\begin{equation}\label{reproducing}
\varphi \star \varphi = \varphi \, ,
\end{equation}
or, in other words, $\varphi$ is an idempotent element in $\mathcal{F}(\mathbf{G})$.   
Finally, given a positive semidefinite function $\varphi$, and given two events $a,b \in \Omega$, will define the transition amplitude $\varphi_{ba}$ as the sum of the amplitudes $\varphi(\alpha)$ for all transitions $\alpha \colon a \to b$:
\begin{equation}
\varphi_{ba} = \sum_{\alpha \colon a \to b} \varphi (\alpha) \, .
\end{equation}
 In other words, we may think of $\varphi_{ba}$ as the amplitude assigned to obtaining the outcome $b$ after having obtained the outcome $a$ by the quantum measure $\mu_\varphi$ associated to the positive semi-definite function $\varphi$, or, equivalently, to the state $\rho_\varphi$ determined by $\varphi$.  
 
The previous definition relates the discussion on the statistical interpretation as quantum measures determined by states on the algebra of the groupoid describing a quantum system with Feynman's phenomenological lodestone discussed in the introduction, that is, Eq. (\ref{feynman}). 
 Then, it is easy to show that:

\begin{proposition}\label{prop:FKC} Let $\varphi$ be an idempotent positive semi-definite function on the (finite) groupoid $\mathbf{G}$, that is, it satisfies the reproducing
property condition Eq. (\ref{reproducing}).
Then, the transition amplitudes $\varphi_{a'a}$ associated to it satisfy Feynman's composition law (that may be called the abstract Chapman-Kolmogorov reproducing equation):
$$
\varphi_{a'a} = \sum_{a''\in \Omega} \varphi_{a'a''} \varphi_{a''a} \, , \qquad \forall a,a' \in \Omega \, .
$$
\end{proposition}

\begin{proof} Clearly, we get:
\begin{equation}\label{comp1}
\varphi_{a'a} = \sum_{\gamma \in \mathbf{G}(a,a')} \varphi (\gamma) =\sum_{\gamma \in \mathbf{G}(a,a')} (\varphi \star \varphi)(\gamma) = \sum_{\gamma \in \mathbf{G}(a,a')} \sum_{\tiny{\begin{array}{c}(\alpha, \beta) \in \mathbf{G}_2 \\ \alpha \circ \beta = \gamma \end{array}}} \varphi(\alpha) \varphi (\beta) \, ,
\end{equation}
where we have used the reproducing property (\ref{reproducing}) for $\varphi$, and the definition of the convolution product.
But now, if $\gamma = \alpha \circ \beta$, because $\gamma \colon a \to a'$, then $\beta \colon a \to a''$ and $\alpha \colon a'' \to a'$ for some $a'' \in \Omega$. Then, the last term in the r.h.s. of Eq. (\ref{comp1}) can be written as:
$$
\sum_{\gamma \in \mathbf{G}(a,a')} \sum_{\tiny{\begin{array}{c}(\alpha, \beta) \in \mathbf{G}_2 \\ \alpha \circ \beta = \gamma \end{array}}} \varphi(\alpha) \varphi (\beta)  = \sum_{a'' \in \Omega}\sum_{\alpha \in \mathbf{G}(a'',a')}  \sum_{\beta \in \mathbf{G}(a,a'')}\varphi(\alpha) \varphi (\beta) = \sum_{a'' \in \Omega} \varphi_{a'a''} \varphi_{a''a} \, .
$$
\end{proof}

We should point out that the transition amplitude $\varphi_{a'a}$ can also be expressed as the transition amplitude associated to the representation of the function $\varphi$ in the space $\mathcal{H}_\Omega$ provided by the fundamental representation $\pi_0$, that is (see \cite{Ib18b}):
$$
\varphi_{a'a} = \langle a' | \pi_0(\varphi) | a \rangle =  \sum_{\gamma \in \mathbf{G}(a,a')} \varphi (\gamma) \, .
$$
Therefore, a simple alternative proof of Prop. \ref{prop:FKC} is obtained by the following computation:
\begin{eqnarray*}
\langle a' | \pi_0(\varphi) | a \rangle &=& \langle a' | \pi_0(\varphi \star \varphi) | a \rangle =  \langle a' | \pi_0(\varphi ) \pi_0(\varphi) | a \rangle \\ &=&  \langle a' | \pi_0(\varphi ) I \pi_0(\varphi) | a \rangle  =  \langle a' | \pi_0(\varphi ) \pi_0 (\mathbf{1}) \pi_0(\varphi) | a \rangle \\ &=&  \sum_{a'' \in \Omega}\langle a' | \pi_0(\varphi ) \pi_0(1_{a''}) \pi_0(\varphi) | a \rangle = 
 \sum_{a'' \in \Omega}  \langle a' | \pi_0(\varphi ) |a''\rangle \langle a'' |\pi_0(\varphi) | a \rangle \, ,
\end{eqnarray*}
where we have used the fact that $\pi_0$ is a representation of the $C^*$-algebra $\mathcal{F}(\mathbf{G})$ so that $\pi_0(\varphi \star \varphi) = \pi_0(\varphi ) \pi_0(\varphi)$, and so that the projectors $\pi_0(1_{a''}) = | a'' \rangle \langle a'' |$ provide a resolution of the identity in $\mathcal{H}_{\Omega}$.


\subsection{Factorizing states}

Generic states are insensitive to the `local' structure of the algebra of transitions codified by the composition law $\alpha \circ \beta$, that is, the amplitudes $\varphi (\alpha \circ \beta)$ are, in general, not directly related to the amplitudes of the factors $\varphi(\alpha)$ and $\varphi(\beta)$.   

However, there is a natural class of states that can be constructed out of the information provided by the factors, that is, states that are characterised in terms of the values of the associated smeared character $\varphi$ on a family of transitions generating the groupoid.  Then, we will say that a state, or the corresponding smeared character $\varphi$, is factorizable if for any pair of composable transitions $(\alpha, \beta) \in \mathbf{G}_2$:
\begin{equation}\label{local}
\varphi (\alpha \circ \beta) = \varphi (\alpha) \varphi (\beta) \, .
\end{equation}
The reversibility of transitions suggest the unitarity preserving property:
\begin{equation}\label{unitarity}
\quad \varphi (\alpha^{-1}) = \varphi(\alpha)^* \, ,
\end{equation}
that will be assumed in addition to the strict factorization property (\ref{local}).   

Notice that condition (\ref{unitarity}) is independent of the factorization condition (\ref{local}) and it can be lifted when dealing with open systems.
Note also that as a consequence of the factorization condition, Eq. (\ref{local}), $\varphi(1_x) = 1$ (because $\varphi(1_x\circ 1_x) = \varphi(1_x)$) and, in addition, $|\varphi(\alpha)| = 1$ because of the unitarity condition, Eq. (\ref{unitarity}).

It is important to remark here that factorizing states do not define (one-dimensional) representations of groupoids.
Even if Eq. (\ref{local}) could give the impression that the function $\varphi \colon \mathbf{G} \to \mathbb{C}$ defines a `linear representation' of the groupoid, this is not so.   
Indeed, a linear representation of as groupoid $\mathbf{G}\rightrightarrows \Omega$ is a functor $R$ from $\mathbf{G}$ to the category of linear spaces $\mathbf{Vect}$, that is, to any outcome $x\in \Omega$ we associate a linear space $V_x = R(x)$ and to any morphism, $\alpha \colon x \to y$, a linear map $R(\alpha) \colon R_x \to R_y$ in such a way that the structure defined by the composition law is preserved (i.e., $R(\alpha \circ \beta) = R(\alpha) R(\beta)$ and $R(1_x ) = \mathrm{id}_{V_x}$).    The simplest possibility would be to associate the 1-dimensional linear space $\mathbb{C}$ to each event $a\in \Omega$ (notice the $R(1_x)$ must be invertible, thus $R(x) \neq \{ \mathbf{0} \}$).  Thus, the total space would have dimension equal to the order of $\Omega$.   Hence, unless $|\Omega| =1$, as it happens in the case of ordinary groups, the smallest possible representation of a groupoid has dimension larger than 1 (such smallest representation is obviously irreducible and is what we have been calling the fundamental representation $\pi_0$ of the groupoid, \cite{Ib19}).

On the other hand, it is easy to see that, in general, a factorizable state $\rho$ does not define a representation of the algebra of the groupoid neither.   If $\rho \colon \mathbb{C}[\mathbf{G}] \to \mathbb{C}$ is the state such that $\rho (\mathbf{a}) = \sum_\alpha a_\alpha \, \varphi (\alpha )$ with $\varphi$ satisfying Eq. (\ref{local}), then it is not true, in general, that $\rho(\mathbf{a}\cdot \mathbf{b})$ agrees with the product $\rho(\mathbf{a} )\rho( \mathbf{b})$.    The reason for this is that, in the evaluation of $\rho(\mathbf{a}\cdot \mathbf{b})$, only the terms $\varphi(\alpha\circ \beta)$ with $\alpha$ and $\beta$ composable will appear, while in $\rho(\mathbf{a}) \rho( \mathbf{b})$ all products $\varphi(\alpha)\varphi(\beta)$ will contribute making the two of them different.

\medskip

Notice that the amplitude $\varphi (\alpha)$ of a factorizable state can always be written as:
\begin{equation}\label{s0}
\varphi (\alpha) = e^{is(\alpha)} \, ,
\end{equation}
for a real-valued function $s\colon \mathbf{G} \to \mathbb{R}$ satisfying the following properties:
\begin{equation}\label{s1}
s(1_x) = 0 \, , \qquad \forall x\in \Omega \, ,
\end{equation}
and
\begin{equation}\label{action}
s(\alpha \circ \beta) = s(\alpha) + s(\beta) \, , 
\end{equation}
for any pair of composable transitions $\alpha$, $\beta$.  Then, we get immediately that $s$ must satisfy:
\begin{equation}\label{s3}
s(\alpha^{-1}) = - s(\alpha) \, .
\end{equation}
We will call a real valued function $s$ on a groupoid satisfying the conditions (\ref{s1}), (\ref{action}), (\ref{s3}), an action.

Even if the discussion of the statistical interpretation of the formalism has been done without reference to any particular dynamics, the structure of factorizable states is strongly reminiscent of Dirac-Feynman definition of amplitudes in the standard space-time interpretation of quantum mechanics, and this is why we will call such function $s$ an action (actually, adding a continuity condition, we may obtain that factorizability implies the existence of a Lagrangian density, closing again Dirac's intuition of the role played by the Lagrangian in Quantum Mechanics, \cite{Di33}).   

We may ask now what properties must an action $s\colon \mathbf{G} \to \mathbb{R}$ possess, beyond those expressed in its definition, for the function $\varphi = e^{is}$ to define a state, that is, to be positive semi-definite, and, in that case, to satisfy the reproducing property. The answer is surprisingly straightforward (and extremely satisfactory): the functions $\varphi (\alpha) = e^{is(\alpha)}$ defined by means of actions are always positive semi-definite, that is, they define states, and those states are always reproducing, that is, they satisfy Feynman condition.

\begin{theorem}\label{local_states} Let $s\colon \mathbf{G} \to \mathbb{R}$ be an action on a finite groupoid $\mathbf{G}$.  Then, the function $\varphi = e^{is}$ is positive semi-definite and satisfies the reproductive property $\varphi = \varphi \star \varphi$.  We will call the state defined in this way the dynamical state of the theory defined by the action $s$.
\end{theorem}

\begin{proof}   Let $n \in \mathbb{N}$, $\xi_i \in \mathbb{C}$ and $\alpha_i\in \mathbf{G}$, $i= 1, \ldots, n$.  

We will prove that $e^{is}$ is positive semidefinite by induction on $n$, i.e., we will show that:
$$
S_n = \sum_{i,j=1}^n \bar{\xi}_i\xi_j e^{is(\alpha_i^{-1}\circ \alpha_j)} \geq 0 \, ,
$$
where only composable pairs $\alpha_i^{-1}\circ\alpha_j$ appear in the expansion of the sum (notice that if $\alpha_i^{-1}$ is composable with $\alpha_j$, then $\alpha_j^{-1}$ is composable with $\alpha_i$),
by complete induction on $n$. 

Thus, if $n =1$, there is only a complex number $\xi$ and a transition $\alpha$, and the sum $S_1 = |\xi|^2\geq 0$ is trivially non-negative.
We will explore also the cases $n = 2,3$ because they will provide the key for the induction argument.  

If $n = 2$, we will be considering complex numbers $\xi_1, \xi_2$ and transitions $\alpha_1$, $\alpha_2 \in \mathbf{G}$.  There will be two possibilities, either $\alpha_1$ and $\alpha_2$ are composable or they are not.  If they are, then we have:
$$
S_2 =  \sum_{i,j=1}^2 \bar{\xi}_i \xi_j e^{is(\alpha_i^{-1}\circ \alpha_j)} =  \sum_{i,j=1}^2 \bar{\xi}_i \xi_j e^{-is(\alpha_i) }e^{is( \alpha_j)} = |\xi_1 e^{is(\alpha_1)} + \xi_2 e^{is(\alpha_2)}|^2 \geq 0 \, ,
$$
while if $\alpha_1^{-1}$ and $\alpha_2$ are not composable then, $S_2 = |\xi_1|^2 + |\xi_2|^2 \geq 0$.  
To understand the general situation we may discuss the case $n = 3$ too.  Then we will have three complex numbers $\xi_1, \xi_2, \xi_3$ and three transitions $\alpha_1$, $\alpha_2$ and $\alpha_3$. There are three cases: all three transitions are composable, two are composable, say $\alpha_1,\alpha_2$ and one is not, and the three are not composable or disjoint.  In the first case a simple computation shows that:
$$
S_3 = |\xi_1 e^{is(\alpha_1)} + \xi_2 e^{is(\alpha_2)} + \xi_3 e^{is(\alpha_3)}|^2 \geq 0\, ,
$$
while in the second and third, we get respectively:
$$
S_3 = |\xi_1 e^{is(\alpha_1)} + \xi_2 e^{is(\alpha_2)}|^2 + |\xi_3|^2\geq 0 \, , \quad S_3 = |\xi_1|^2 + |\xi_2|^2 + |\xi_3|^2 \geq 0 \, .
$$
Let us consider $n$ arbitrary, then the relation $i \sim j$ if $\alpha_i^{-1}$ is composable with $\alpha_j$ (or in other words, if the targets of $\alpha_i$ and $\alpha_j$ are the same, $t(\alpha_i) = t(\alpha_j)$) is an equivalence relation on the set of indices $I_n = \{ 1, 2, \ldots, n\}$.   The set $I_n$ is decomposed into equivalence classes $I_{x} = \{ i_{x_1}, \ldots, i_{x_r} \}$ that will correspond to all transitions $\alpha_i$ such that $t(\alpha_i) = x$, and each class will have a number of elements $n_x \leq n$.   Then, if $n_x = n$, there is only one class, all pair of transitions $\alpha_i^{-1}$, $\alpha_j$ are composable and then:
$$
S_n = \big|\sum_{k=1}^n \xi_k e^{is(\alpha_k)} \big|^2 \geq 0 \, .
$$
On the other hand, if there is more than one equivalence class, then $n_x < n$ for all $x \in \Omega$, and   we have:
$$
S_n = \sum_{x\in \Omega} \sum_{j_x,k_x \in I_x} \bar{\xi}_{j_x}\xi_{k_x} e^{-is(\alpha_{j_x})}e^{is(\alpha_{k_x})} = \sum_{x \in \Omega} \big| \sum_{k_x\in I_x} \xi_{k_x} e^{is(\alpha_{k_x})} \big|^2 \geq 0 \, , 
$$
where in the last step in the previous computation we have used the induction hypothesis.  This shows that $\varphi = e^{is}$ is positive semidefinite.  

\medskip

To prove the reproducing property, we normalise the smeared character $\varphi$ properly as:
$$
\varphi (\alpha) = \frac{|\Omega|}{|\mathbf{G}|} e^{is(\alpha)} \, .
$$
Then, a simple computation shows that:
\begin{eqnarray}
\varphi \star \varphi (\gamma) &=& \frac{|\Omega|^2}{|\mathbf{G}|^2}  \sum_{\tiny{\begin{array}{c}(\alpha, \beta) \in \mathbf{G}_2 \\ \alpha \circ \beta = \gamma \end{array}}} e^{is(\alpha)} e^{is(\beta)} \nonumber \\ &=& \frac{|\Omega|^2}{|\mathbf{G}|^2} \sum_{\tiny{\begin{array}{c}(\alpha, \beta) \in \mathbf{G}_2 \\ \alpha \circ \beta = \gamma \end{array}}} e^{is(\alpha\circ \beta)} \nonumber \\ &=&  \frac{|\Omega|}{|\mathbf{G}|}  e^{is(\gamma)}  = \varphi (\gamma)\, , \label{sumclass}
\end{eqnarray}
where, in the step (\ref{sumclass}) in the previous computation, we have used that the argument of the sum, $e^{is(\alpha\circ \beta)}$, is constant and equal to $e^{is(\gamma)}$ whenever $\alpha \circ \beta = \gamma$, but because the number of composable transitions $\alpha$, $\beta$ such that $\alpha\circ \beta = \gamma$ is exactly $|\mathbf{G}|/|\Omega|$, then we get the required factor and the conclusion.  

Let us justify this last statment.  First, notice that, if $\gamma \colon x \to y$, then for any $\alpha \colon x \to z$ there is exactly one $\beta = \alpha^{-1}\circ \gamma$ such that $\alpha \circ \beta = \gamma$.
Consequently, the number of pair transitions factorising $\gamma\colon x \to y$ is $|\mathbf{G}_+(x)|$, but we also have $\sqcup_{x\in \Omega} \mathbf{G}_+(x) = \mathbf{G}$, which means $|\mathbf{G}| = |\Omega|  |\mathbf{G}_+(x)|$ and the statement is proved.
\end{proof}


 We can summarise all previous discussion by saying that we can understand the description of a quantum system in the groupoid formalism (which provides an abstraction of Schwinger algebra of measurements) as a grade-2 measure theory provided by an invariant quantum measure $\mu$.   Such quantum measure is characterised by a positive semi-definite function $\varphi$ on the groupoid, and for any action function $s$ on the groupoid, the function $\varphi = e^{is}$ is positive semi-definite, is factorizable, it  satisfies the reproducing property, and defines uniquely a quantum measure $\mu_s$ whose decoherence functional $D_s$ has Sorkin's form:
 $$
 D_s(\alpha, \beta) = e^{-iS(\alpha)} e^{iS(\beta)} \delta(t(\alpha), t(\beta)) \, .
 $$
 Here, $\alpha$ and $\beta$ denote two transitions in the groupoid $\mathbf{G}$ and we have made explicit the delta function of the targets.



\section{The statistical interpretation of Schwinger's transformation functions}\label{sec:formalism}

In the previous sections, it was discussed how the notion of state on the $C^*$-algebra of a quantum system described by a groupoid $\mathbf{G}\rightrightarrows \Omega$ provides a statistical interpretation of the theory in terms of Sorkin's notion of quantum measure and the theory of decoherence functionals, and clarifies the origin of Feynman-Dirac's amplitudes and their reproducing property. 

In this section, as anticipated in the introduction, we will provide a natural statistical interpretation of Schwinger's transformation functions by relying again on the key notion of states. This time, we will provide a natural interpretation of transition amplitudes on the fundamental representation of a given groupoid by using particularly simple states.  Moreover, a judiciously use of the the fundamental invariance of the description of the system with respect to changes of systems of observables will provide the desired interpretation.    


\subsection{Equivalence of algebras of observables}\label{sec:transitions}

It is a fundamental assumption of the theory developed so far that if we select a compatible set of observables $\mathbf{A}$ for the system, the algebra of observables of the system will contain the $C^*$-algebra of the goupoid $\mathbf{G}_\mathbf{A}$ determined by the system $\mathbf{A}$ \cite{Ib18b}.     The groupoid $\mathbf{G}_\mathbf{A}$ will consist of all possible physical transitions $\alpha \colon a \to a'$ among events $a \in \Omega_A$ determined by the set of compatible observables $\mathbf{A}$.   

Given a description provided by the groupoid $\mathbf{G}_\mathbf{A}$, we may consider a finer description of the system by using another groupoid $\mathbf{G}_{\mathbf{A}'}$ such that $\mathbf{G}_\mathbf{A} \subset \mathbf{G}_{\mathbf{A}'}$ is a subgroupoid.   This will imply that the physical description of the system provided by $\mathbf{G}_\mathbf{A}$ is consistent with the physical description of the system provided by $\mathbf{G}_{\mathbf{A}'}$.   Then, the corresponding groupoid algebras will satisfy $\mathbb{C}[\mathbf{G}_\mathbf{A}] \subset \mathbb{C}[\mathbf{G}_{\mathbf{A}'}]$, and the $C^*$-algebras of observables provided by both descriptions will  be related accordingly.   In this sense, we will say that a description of a quantum system is complete if the algebra of observables $\mathcal{A}_\mathbf{G} = C^*(\mathbf{G})$ provided by the groupoid $\mathbf{G} \rightrightarrows \Omega$ is maximal.

Let us suppose that the groupoids $\mathbf{G}_\mathbf{A} \rightrightarrows \Omega_A$ and $\mathbf{G}_\mathbf{B} \rightrightarrows \Omega_B$ provide complete descriptions of the same quantum system.  It is just natural to assume that the corresponding algebras of observables $C^*(\mathbf{G}_\mathbf{A} )$ and $C^*(\mathbf{G}_\mathbf{B} )$ are isomorphic because if this were not the case there would be physical states that could be obtained in one description but not in the other.  In other words, a complete description of the system cannot depend on the choice of a particular set of compatible observables.   Then, there would be an isomorphism of $C^*$-algebras:  
$$
\tau_{AB} \colon C^*(\mathbf{G}_\mathbf{A}) \to C^*(\mathbf{G}_\mathbf{B}) \, ,
$$ between the corresponding $C^*$-algebras in both `reference frames' $\mathbf{A}$ and $\mathbf{B}$.   This independence of the description with respect to the chosen `reference frame' was stated as a `relativity principle' in \cite{Ib18b} that will be developed in what follows.   

Together with the isomorphism $\tau_{AB}$, there is an isomorphism $ \tau_{BA} \colon C^*(\mathbf{G}_\mathbf{B}) \to C^*(\mathbf{G}_\mathbf{A})$, and then it is natural to conclude that\footnote{Notice that a categorical approach to these notions will just impose that $\tau_{AB}$ and $\tau_{BA}^{-1}$ would differ on an automorphism of the underlying algebras, however, we will just consider the strict interpretation of equalities here or, in other words that the categorical notions behind the structures we are dealing with are defined in the strong sense.} $\tau_{AB} = \tau_{BA}^{-1}$.
In the same vein if $\mathbf{C}$ is another complete system of observables yet, then there will exists isomorphisms of $C^*$-algebras $\tau_{BC} \colon C^*(\mathbf{G}_\mathbf{B}) \to C^*(\mathbf{G}_\mathbf{C})$ and $\tau_{AC} \colon C^*(\mathbf{G}_\mathbf{A}) \to C^*(\mathbf{G}_\mathbf{C})$, that will be assumed to satisfy the natural composition law:
$$
\tau_{BC} \circ \tau_{AB} = \tau_{AC} \, .
$$


\subsection{Transition amplitudes again}

Finally, let us recall (see \cite{Ib18b}) that an observable is a function $f \in \mathcal{F}(\mathbf{G}_\mathbf{A}) \subset C^*(\mathbf{G}_\mathbf{A})$ such that $f^* = f$, that is, a self-adjoint element in the $C^*$-algebra of the groupoid.
We will define the \textit{transition amplitude} of the observable $f$ between two events $a$ and $a'$
as the sum\footnote{We will assume that the groupoid is finite, if not, obvious changes in the formulas replacing sums by integrals with respect to properly chosen measures should be introduced.} of the values of the observable over all transitions connecting $a$ and $a'$ and we will denote it by $\langle a' ; f ;  a \rangle$:
$$
\langle a' ; f ;  a \rangle = \sum_{\alpha \in \mathbf{G}(a',a)} f(\alpha) \, .
$$
Notice that 
$$
\langle a' ; f^* ;  a \rangle  =    \sum_{\alpha \in \mathbf{G}(a',a)} f^*(\alpha) =  \sum_{\beta \in \mathbf{G}(a,a')} \bar{f}(\beta ) = \overline{\langle a' ; f ;  a \rangle} \, ,
$$
and if we denote by $\langle a ;  a' \rangle$ the amplitude corresponding to the unit $\mathbf{1}$, that is, $\langle  a ;  a' \rangle = \langle a ; \mathbf{1} ;a' \rangle$, then:
$$
\langle a ;  a' \rangle = \delta(a,a') \, ,
$$
because 
$$
\langle a ;  a \rangle = \langle a ; \mathbf{1} ;a \rangle =  \sum_{a'\in \Omega}\langle a ; \delta_{a'} ;a \rangle = \sum_{a'\in \Omega} \sum_{\alpha \in \mathbf{G}(a,a)} \delta_{a'}(\alpha) = 1 \, ,
$$
and $\langle a ;  a' \rangle = 0$, if $a \neq a'$, as $\mathbf{1} = \sum_{a \in \Omega} \delta_a$ must be evaluated on transitions $\alpha$ with different source and target. 

Another interesting observable is provided by the `incidence matrix' observable $\mathbb{I} = \sum_{\alpha \in \mathbf{G}} \delta_\alpha$.   Notice that $\mathbb{I}^* = \mathbb{I}$ and:
$$
\langle a'; \mathbb{I}; a \rangle = \sum_{\alpha \in \mathbf{G}(a',a) } \mathbb{I}(\alpha) = |\mathbf{G}(a',a) | \, .
$$
It is also relevant to point out the if $\varphi$ is a positive semi-definite function on $\mathbf{G}$ then $\varphi^* = \varphi$ (notice that because $\sum_{i,j=1}^n \bar{\xi}_i \xi_j \varphi (\alpha_i^{-1}\circ \alpha_j) \geq 0$ for all $\xi_i$, then $\overline{\varphi (\alpha_i^{-1}\circ \alpha_j)} = \varphi (\alpha_j^{-1}\circ \alpha_i)$ for all composable $\alpha_i^{-1}$, $\alpha_j$, but then it holds for all $\alpha$), and 
$$
\langle a'; \varphi ; a\rangle =   \sum_{\alpha \in \mathbf{G}(a',a) } \varphi (\alpha) = \varphi_{a'a} \, , 
$$
and the transition amplitude $\langle a'; \varphi ; a\rangle$ is just the transition amplitude of the state $\varphi$ considered in Sect. \ref{sec:reproducing}.


\subsection{The states $\rho_x$ and their associated GNS constructions}\label{sec:rho_x}

To relate the definition of transition amplitudes with the standard interpretation of such functions in terms of vector-states and operators, and eventually with Schwinger's transformation functions, we have to select a representation of the theory.   

As discussed in Sect. \ref{sec:GNS}, the representations of the $C^*$-algebra $\mathbb{C}[\mathbf{G}]$ are defined via the GNS construction.  Hence, following the spirit so far, we will choose a particular state that will provide a particular representation of transition amplitudes.   For that, and as a further illustration of the GNS construction, we will consider the simple state $\rho_x$ defined by the function $\delta_x$, that is, $\rho_x(\mathbf{a}) = a_x$ where $\mathbf{a} = \sum_{\alpha} a_\alpha \, \alpha$, that is, $\rho_x$ assigns to any virtual transition $\mathbf{a}$ the coefficient of the unit $1_x$.   Clearly $\rho_a(\mathbf{1}) = 1$ and 
\begin{equation}\label{rhox}
\rho_x(\mathbf{a}^*\cdot \mathbf{a}) = \sum_{\alpha \in \mathbf{G}_+(x)} | a_\alpha |^2 \geq 0 \, ,
\end{equation}
that shows that $\rho_x$ is indeed a state.

Following the GNS construction described in Sect. \ref{sec:GNS} (see also \cite[Sect. 4]{Ib18b}), we see that
the Hilbert space $\mathcal{H}_{\rho_x}$, denoted in what follows by $\mathcal{H}_x$, is the Hilbert space of functions $\Phi$ defined on $\mathbf{G}_+(x)$ with the standard inner product.    In fact, from Eq. (\ref{rhox}) we see that the Gelfand ideal $\mathcal{J}_x = \{ \mathbf{a} \mid \rho_x(\mathbf{a}^*\cdot \mathbf{a})  = 0 \}$ consists of all $\mathbf{a}$ such that the coefficients of transitions $\alpha \in \mathbf{G}_+(x)$ vanish.  That means that the quotient space $\mathbb{C}[\mathbf{G}]/\mathcal{J}_x$ can be identified with the space of transitions in $\mathbf{G}_+(x)$, and thus, given any $\mathbf{a}\in \mathbb{C}[\mathbf{G}]$, we will use the notation $\mathbf{a}_x$ for the restriction to $\mathbf{G}_+(x)$, i.e., $\mathbf{a}_x$ is obtained from $\mathbf{a}$ by putting to zero all coefficients $a_\alpha$ with $\alpha \notin \mathbf{G}_+(x)$ or, in other words, $\mathbf{a}_x = \mathbf{a}\cdot 1_x$.  Moreover, the  inner product $\langle \cdot, \cdot \rangle_x$ in $\mathcal{H}_x$ induced by $\rho_x$ is given by, Eq. (\ref{inner_rho}):
\begin{equation}\label{inner_x}
\langle \mathbf{a}_x, \mathbf{a}'_x \rangle_x = \rho_x (\mathbf{a}^* \cdot \mathbf{a}') =  \sum_{\alpha \in \mathbf{G}_+(x)} \bar{a}_\alpha a'_\alpha \, .
\end{equation}

In particular, the unit $\mathbf{1}$ determines the fundamental vector $ \mathbf{1}_x = 1_x \in \mathcal{H}_x$.   The algebra $\mathbb{C}[\mathbf{G}]$ is represented in $\mathcal{H}_x$ as $\pi_x(\mathbf{a}) \mathbf{a}'_x = (\mathbf{a}\cdot \mathbf{a}')_x = \mathbf{a}\cdot \mathbf{a}'_x$, and clearly $1_x$ is a cyclic vector for such representation.  Now, instead of denoting by $|0\rangle$ the ground vector of the representation $\pi_x$, for convenience, we will denote it by $|x\rangle$. Thus, if $\mathbf{a}$ is a virtual transition, we have:
$$
\pi_x(\mathbf{a}) |x\rangle = \mathbf{a}_x \, .
$$
In order to have a homogeneous notation, we can write as $\mathbf{a}_x = |\mathbf{a}\rangle_x$ where the subscript $x$ indicates that the vector $|\mathbf{a}\rangle_x$ belongs to the Hilbert space $\mathcal{H}_x$. Thus, using this notation in Eq. (\ref{inner_x}), we have:
$$
\langle \mathbf{a}_x, \mathbf{a}'_x \rangle =  \sum_{\alpha \in \mathbf{G}_+(x)} \bar{a}_\alpha a'_\alpha = \langle \mathbf{a} \mid \mathbf{a}' \rangle_x \, ,
$$
which is the convenient form of expressing the inner product that will be used in the following.  With this notation, the amplitude defined by the state $\rho_x$ on a virtual transition $\mathbf{a}$ can be written as (recall Eq. (\ref{amplitude_phi})):
\begin{equation}\label{rho_xa}
\rho_x (\mathbf{a}) = \langle x \mid  \mathbf{a} \rangle_x \, .
\end{equation}


\subsection{Transformation functions and transition amplitudes}

We are ready to interpret Schwinger's transformation functions $\langle b | a \rangle$ as transition amplitudes and hence to provide them with a proper statistical interpretation.   Let us recall that, according to Schwinger, the transformation function $\langle b | a \rangle$ ``is a number characterising the statistical relation relation between the states $b$ and $a$'', and reflects the fact ``that only a determinate fraction of the systems emerging from the first stage will be accepted by the second stage''.  

In the formalism we have developed, Schwinger's transformation function will be given by the isomorphism $\tau_{AB}$ that relates the $\mathbf{A}$ and $\mathbf{B}$ descriptions of the system, and we would like to provide a  statistical interpretation of the complex number $\langle b|a\rangle$ appearing in Schwinger's formalism as transition amplitude.  For that, consider that in the description provided by the complete family $\mathbf{B}$ of observables we want to understand the statistical relation between the outcome $b$, i.e., the transition $1_b$ in the algebra $\mathbb{C}[\mathbf{G}_\mathbf{B}]$, and the transition $1_a$ corresponding to the outcome $a$ with respect to the description provided by the family $\mathbf {A}$, that is, the algebra $\mathbb{C}[\mathbf{G}_\mathbf{A}]$. Then, such relation is provided by the amplitude of the state $\rho_a$ defined by $a$ in $\mathbb{C}[\mathbf{G}_\mathbf{A}]$ on the transition defined by $\tau_{BA}(1_b) \in \mathbb{C}[\mathbf{G}_\mathbf{A}]$.  But then, using Eq. (\ref{rho_xa}), we get:
$$
\rho_a (\tau_{BA}{1_b}) = \langle a | \tau_{BA}(1_b) \rangle_a \, .
$$
If we denote the vector state in the Hilbert space $\mathcal{H}_a$ defined by the transition $\tau_{BA}(1_b)$ by $|b\rangle$, that is:
$$
| b \rangle = \pi_a(\tau_{BA}(1_b))|a\rangle \, ,
$$ 
we get that the transition amplitude of the event $b$ with respect to the state defined by $a$, that we may denote consistently as $\varphi_{ba}$, is given by:
$$
\varphi_{ba} = \langle b | a \rangle \, .
$$
Notice that the we could have proceeded the other way around, exchanging the roles of $a$ and $b$, and then, repeating the argument, we get that the transition amplitude $\varphi_{ab}$ of the event $a$ with respect to the state defined by $b$, would have been:
$$
\varphi_{ab} = \langle a | b \rangle = \overline{\langle b | a \rangle} = \overline{\varphi}_{ba} \, .
$$
Notice that the previous identities follow from the duality of states and transitions and the properties of the isomorphisms $\tau_{AB}$, that is:
$$
\rho_a ( \tau_{BA} (1_b)) = \rho_{\tau_{AB}(1_a)}(1_b) = \rho_b (\tau_{AB}(1_a)) \, , \qquad \forall a \in \Omega_A, b \in \Omega_B \, .
$$


\section{Some simple applications:  the qubit and the two-slit experiment}

\subsection{The qubit}  We can illustrate the ideas discussed along this paper by using the qubit system.   The qubit system is the simplest nontrivial quantum system and in the groupoid formalism correspond to the groupoid defined by the graph $A_2$, that is, the space  of outcomes $\Omega = \{ + , - \}$ consists of two events $+$, $-$, and there is one non-trivial transition $\alpha \colon - \to +$.   In addition to this, there are two units $1_\pm$ and the inverse $\alpha^{-1}\colon + \to -$ of the transition $\alpha$, with $\alpha^{-1}\circ \alpha = 1_-$, $\alpha \circ \alpha^{-1} = 1_+$ (see Fig. \ref{A2}).   This scheme abstracts the simplest situation of a physical system evolving in time and producing two outcomes denoted by $+$ and $-$.  

\begin{figure}[h]
\centering
\begin{tikzpicture} 
\fill (0,0) circle  (0.1);
\fill (3.5,0) circle  (0.1);
\draw [thick,->] (0.3,0.1) arc (107:73:5);
\draw [thick,->] (3.2,-0.1) arc (288:253:5);
\node [left]  at  (-0.1,0)   {$+$};
\node [right] at (3.9,0)   {$-$};
\node [above] at (2,0.6)   {$\alpha$};
\node [below] at (2,-0.6)   {$\alpha^{-1}$};
\end{tikzpicture}
\caption{The abstract qubit, $A_2$.}
\label{A2}
\end{figure}
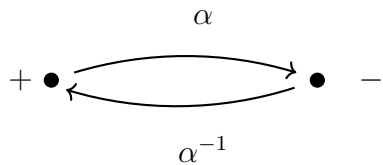

The corresponding groupoid will be denoted by $\mathbf{A}_2$ again and its algebra $\mathbb{C}[\mathbf{A}_2] = \{ \mathbf{a} = a_+ 1_+ + a_- 1_- + a_\alpha \alpha + a_{\alpha^{-1}} \alpha^{-1} \mid a_\pm, a_\alpha, a_{\alpha^{-1}} \in \mathbb{C}\}$ is easily seen to be isomorphic to the algebra $M_2(\mathbb{C})$ of $2\times 2$ complex matrices.  The identification is provided by the assignments:
$$
1_+ \mapsto \left[ \begin{array}{cc} 1 & 0 \\ 0 & 0 \end{array} \right] \, , \quad 1_- \mapsto  \left[ \begin{array}{cc} 0 & 0 \\ 0 & 1 \end{array} \right] \, ,\quad 
\alpha \mapsto  \left[ \begin{array}{cc} 0 & 0 \\ 1 & 0 \end{array} \right] \, , \quad \alpha^{-1} \mapsto  \left[ \begin{array}{cc} 0 & 0 \\ 1 & 0 \end{array} \right]  \, .
$$
Then, a virtual transition $\mathbf{a}$ is associated to the matrix:
$$
A = \left[ \begin{array}{lc} a_+ & a_\alpha \\ a_{\alpha^{-1}} & a_- \end{array} \right] \, , 
$$
and $a^*$ is associated to the matrix $A^\dagger$.  The $C^*$ norm $||\cdot ||$ is just the matrix operator norm and the fundamental representation $\pi_0$ of the algebra becomes the natural defining representation of $M_2(\mathbb{C})$ on $\mathbb{C}^2$.   The vectors associated to the unit elements $1_\pm$ are given by:
$$
| + \rangle = \left[ \begin{array}{c} 1 \\ 0 \end{array} \right] \, , \qquad | - \rangle = \left[ \begin{array}{c} 0 \\ 1 \end{array} \right] \, ,
$$ 
and thus an arbitrary vector in $\mathcal{H}_2 = \mathbb{C}^2$ is written as $|\psi \rangle = \psi_+ |+\rangle + \psi_- | - \rangle$. 

The space of states of the groupoid algebra $\mathbb{C}[\mathbf{A}_2]$ can be identified with the space of density operators on $\mathcal{H}_2$, that is, normalized non-negative, self-adjoint operators $\hat{\rho}$ on $\mathcal{H}_2$.   Density operators can be parametrized as:
$$
\hat{\rho} = \frac{1}{2} (\mathbb{I} - \mathbf{r}\cdot \boldsymbol{\sigma}) \, ,
$$
with $ \mathbf{r} \in \mathbb{R}^3$ a vector in Bloch's sphere, $r = || \mathbf{r} || \leq 1$, and $ \boldsymbol{\sigma} = (\sigma_1,\sigma_2, \sigma_3)$, the standard Pauli matrices.

According to Thm. \ref{local_states}, factorizable states have the form $\varphi = e^{is}$, with $s$ and action function.  Then, let $s \colon \mathbf{A}_2 \to \mathbb{R}$ given by:
$$
s(1_\pm ) = 0 \, , \qquad s(\alpha) = - s(\alpha^{-1}) = S \, ,
$$
with $S$ a real number.
Clearly, the function $s$ defined in this way satisfies the additive property (\ref{action}) and the state defined by $\rho_S(\mathbf{a}) = \sum_{i,j} \bar{a}_i a_j \varphi (\alpha_i^{-1}\circ \alpha_j)$ is a factorizable (and reproducing) state.  
The characteristic function $\varphi_s$ defined by the action $s$ is given by:
$$
\varphi_s (1_\pm ) = 1 \, , \qquad \varphi_s (\alpha) = \overline{\varphi_s (\alpha^{-1})} = e^{-iS} \, ,
$$ 
and the associated state $\rho_s$ is given by:
$$
\hat{\rho}_s = \frac{1}{2} \left[ \begin{array}{cc} 1 & e^{-iS} \\ e^{iS} & 1 \end{array} \right] \, .
$$
Notice that $\hat{\rho}_s \hat{\rho}_s = \hat{\rho}_s$, thus it satisfies the reproducing property (it can also be checked directly that $\varphi_s \star \varphi_s = \varphi_s$).

The decoherence functional defined by this state is given by the $4\times 4$ matrix $D_s$ whose entries $(i,j)$ correspond to the values $D_s(\alpha_i, \alpha_j) = \frac{1}{4} \varphi_s (\alpha_i^{-1}\circ \alpha_j) \delta(t(\alpha_i),t(\alpha_j))$, with $\alpha_i$ running through the list $1_+,1_-, \alpha, \alpha^{-1}$.  Thus, for instance, $D_s (1_+,1_+) = \frac{1}{2} \varphi_s (1_+^{-1}\circ 1_+) = 1/4$, $D_s (1_+,1_-) = \frac{1}{2} \varphi_s (1_+^{-1}\circ 1_-) =0,$ and so on.  Therefore, we finally get:
$$
D_s = \frac{1}{4} \left[ \begin{array}{cccc} 1 & 0 &  e^{-iS} & 0  \\ 0 & 1 & 0 & e^{iS} \\  
e^{iS} & 0 & 1 & 0 \\ 0 & e^{-iS} & 0 & 1 
\end{array} \right] \, .
$$
As it was discussed in the main text, the decoherence functional describes the structure of the quantum measure $\mu_s$, and hence the statistical interpretation associated to the system $\mathbf{A}_2$ in the state $\rho_s$.


\subsection{The double slit experiment}

In order to understand better some of the implications of the previous discussion, it is revealing to compare the qubit system with the double slit experiment. For the purposes of the present paper, we will use the analysis of the double slit experiment carried on in \cite{Ga09} in the coarse-graining histories description\footnote{After reading this, it should be clear that an analysis following similar arguments could be performed for the $n$-slit experiment or more complicated systems like Kochen-Specker system \cite[Chap. 2]{Ga09}.}.  We will reproduce succinctly the argument in \cite{Ga09} in order to facilitate the comparison with the previous results.

\begin{figure}[h]
\centering
 \resizebox{12cm}{5cm}{\includegraphics{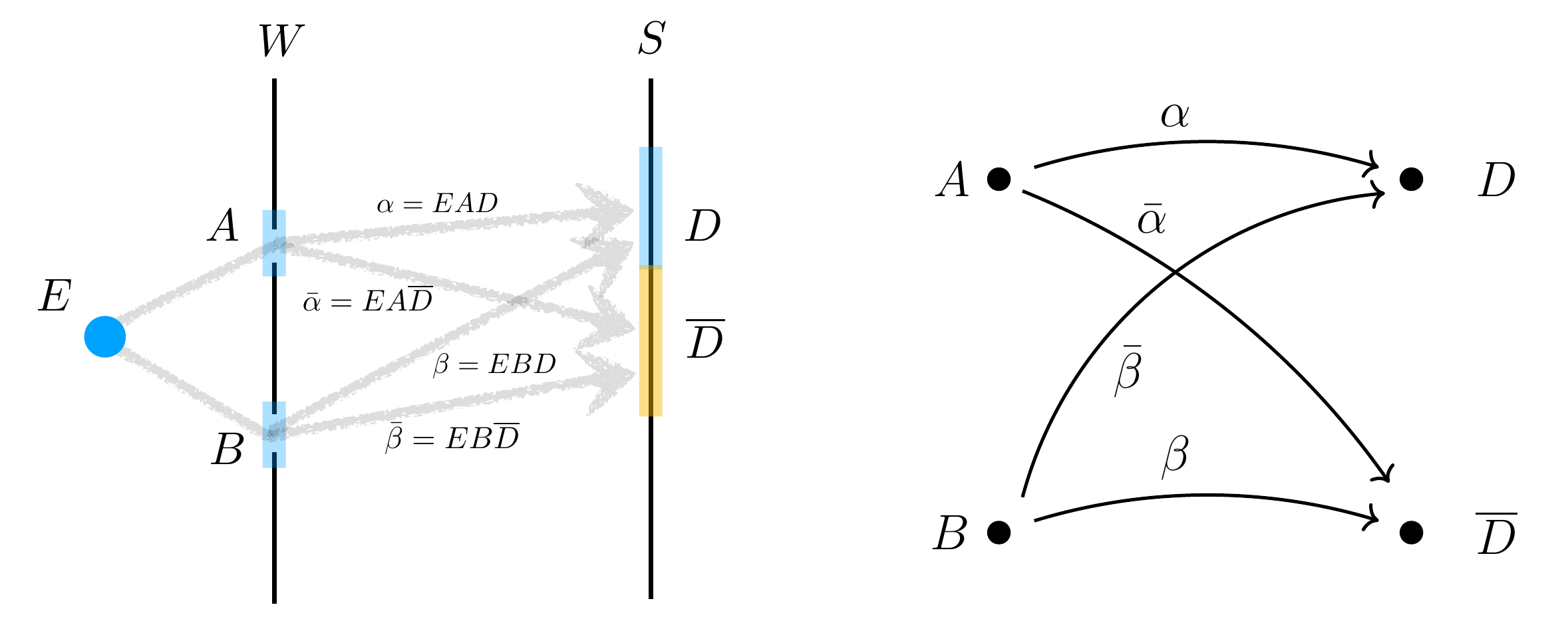}} 
\caption{A coarse graining interpretation of the double slit experiment (left). The corresponding double slit quiver: $U_{2}$ (right).}
\label{U2}
\end{figure}

Consider an idealised double slit system as sketched in Fig. \ref{U2} (left), where a particle is fired from an emitter $E$ and can pass through slits $A$ or $B$ on a wall $W$ before ending on the final screen $S$ either at the detector $D$ (which is located for instance on a dark fringe) or elsewhere, $\overline{D}$.  

In \cite{Ga09} the interpretation of the system is provided in terms of a Hilbert space, initial vector state $|\Psi \rangle$ and projectors $P_A$ corresponding to finding the particle at slit $A$. The projectors $P_B$, $P_D$ are defined in a similar fashion an $P_{\overline{D}} = I - P_D$.   We are not interested in such analysis here as we want to provide an algebraic description of it in terms of the structures discussed in the groupoid formalism.   For that, we will identify a family of `outcomes' given by $A$ and $B$ that will correspond to the particle passing through the slits $A$ or $B$ respectively, and `outcomes' $D$ and $\overline{D}$ corresponding to the particle hitting the region $D$ or $\overline{D}$ on the screen.  Hence, the space of events in this coarse-grained description of the system is finite and has 4 elements, that is, $\Omega = \{ A,B,D,\overline{D}\}$.  

Note that in this picture the notion of outcome/event is not related to a complete family of compatible measurements.  We do not even assume that there are actual detectors at the slits, but we are considering that it would be possible to determine that the particle is located near $A$ precisely enough to discard that it would be close to $B$ and conversely\footnote{This conception of events are closely related to the idealized notion of measurement used by Feynman to talk about trajectories of particles and to the notion of events in Sorkin's approach.}

The physical transitions of the system include the histories (see Fig. \ref{U2}) $\alpha = EAD$, indicated by  the transition $\alpha \colon A \to D$ such that the fired particle causes the event $A$ and consecutively $D$; $\beta = EBD$, that is, the transition $\beta\colon B \to D$ representing the histories that cause the event $B$ and then $D$.  Apart from $\alpha$ and $\beta$, there are two more transitions $\bar{\alpha} = EA\overline{D} \colon A \to \overline{D}$, $\bar{\beta} = EB\overline{D} \colon B \to \overline{B}$ with similar meaning (notice that $\bar{\alpha} \neq \alpha^{-1}$).  The collection of transitions $U_2 = \{ \alpha, \bar{\alpha}, \beta, \bar{\beta} \}$\footnote{The notation $U_2$ corresponds to the notion of `utility' graph used in graph theory.} do not define a groupoid but rather a quiver (see Fig. \ref{U2} for the pictorical representation of it) and they correspond to the family of coarse-grained histories in the description of \cite{Ga09}.    From this point of view, note that it was not necessary to consider the event $E$ since all relevant physical transitions assume that the particle has been fired.

The quiver $U_2$ generates a groupoid $\mathbf{G}(U_2)$ by adding the units $1_A, 1_B, 1_D, 1_{\overline{D}}$, the inverses $\alpha^{-1}, \bar{\alpha}^{-1}, \beta^{-1}, \bar{\beta}^{-1}$ and four more transitions corresponding to $\gamma_{AB} = \beta^{-1}\circ \alpha \colon A \to B$, $\gamma_{D\overline{D}} = \overline{\beta}^{-1}\circ \beta \colon D \to \overline{D}$, etc. (see the elementary introduction to the theory of groupoids and their representations in \cite{Ib19}).   Thus, the order of the groupoid $\mathbf{G}(U_2)$ is 16 and it can be identified with the groupoid of pairs of $\Omega$.  Of course, we may argue about the physical meaning of the transitions $\gamma_{AB}$, $\gamma_{D\overline{D}}$, and so on, as well as on the physical meaning of the inverses $\alpha^{-1}$, $\beta^{-1}$, and so on.  There are no physical reasons to exclude them.  Feynman's microscopic reversibility principle implies the consideration of the inverse transitions $\alpha^{-1}$, etc., in the analysis of the system and then, because of logical consistency, of the transitions $\gamma_{AB}$, $\gamma_{D\overline{D}}$, etc.   There is however no reason to consider states of the system where such transitions could actually happen, that is, they can be precluded so that the quantum measure describing the statistical properties of the system takes the value zero on them. This is exactly the point of view that we will take in our analysis.   We will construct various states of the system possessing this property.

The construction of a quantum measure on $\mathbf{G}(U_2)$ considered as a coarse-grained histories description of the actual system is associated to a state on the algebra $\mathbb{C}[\mathbf{G}(U_2)]$ of the groupoid.  In particular, factorizable states, which are the ones that lead to a dynamical interpretation of the theory, have associated characteristic functions $\varphi \colon \mathbf{G}(U_2) \to \mathbb{C}$ of the form: $\varphi  = e^{is}$ (up to a normalization factor), with $s$ an action functional on the groupoid.   In our case, because $\mathbf{G}(U_2)$ is generated by the utility quiver $U_2$,  it suffices to give the values of $s$ on the transitions $\alpha, \bar{\alpha}, \beta, \bar{\beta}$, that is, in the histories $EAD$, $EA\overline{D}$, $EBD$, and $EB\overline{D}$.   Thus, we may assume that:
$$
s(\alpha ) = s(\beta)  + \delta= S_1 \, , \qquad s(\bar{\alpha}) = s(\bar{\beta}) = S_2 \, ,
$$
where $\delta$, is a phase related to the difference between the physical paths when the particle follows the trajectories $EAD$ and $EBD$ respectively.   A similar phase could be introduced in the action for $\bar{\alpha}$ and $\bar{\beta}$, however, because of the particular configuration of the experiment, they have been chosen to be equal.

Notice that the values of $s$ in all other transitions are determined by the properties of action functionals, for instance, $s(1_A) = 0$, $s(\alpha^{-1}) = - S_1$, and so on.   In particular, $s(\gamma_{BA}) = s(\beta^{-1}\circ \alpha) = \delta$. Thus, the decoherence functional defined by the characteristic function $\varphi$ is codified in a $16\times 16$ matrix $D$ whose entries are given by the numbers $D(\alpha_i,\alpha_j)$.    If we concentrate ourselves in the $4\times 4$ submatrix $D_{U_2}$ corresponding to the quiver $U_2$, that is, corresponding to the transitions $\alpha, \bar{\alpha}, \beta$ and $\bar{\beta}$, we will get:
$$
D_{U_2} = \frac{1}{16} \left[ \begin{array}{cccc} 1 & e^{i\delta} &  0 & 0  \\ e^{-i\delta} & 1 & 0 & 0 \\  
0 & 0 & 1 & 1 \\ 0 & 0 & 1 & 1 
\end{array} \right] \, .
$$

Thus, in the particular instance $\delta = \pi$, we will have the matrix describing the quantum measure $\mu$ associated to the standard double slit experiment interpretation.  Notice that in such case the measure of the set $\mathbf{V} = \{ \alpha = EAD, \beta = EBD \}$ is given by:
$$
\mu(\mathbf{V}) = D_{U_2}(\mathbf{V}, \mathbf{V}) = \sum_{\alpha_1,\alpha_2 \in \mathbf{V} } D_{U_2}(\alpha_1, \alpha_2) = 0 \, ,
$$
in accordance with the fact the the detector is in a dark fringe, that is, the arrival of particles to it is precluded.

Notice that we may change the outcomes, either by moving the wall or the detectors, that is, we modify the outcomes $A',B',D',\overline{D}'$ and the transitions $\alpha', \beta',\bar{\alpha}', \bar{\beta}'$. Then the description of state will change accordingly.   Notice that there will be an isomorphism $\tau$ from the group algebra of the original groupoid $\mathbf{G}(U_2)$ and the one obtained by using the primed data.

Notice that the transitions ending at $D$ have no interference with histories ending at $\overline{D}$. This is a general feature of the groupoid formulation and is due to the composability condition among them.  More sophisticated experiments can be easily analyzed using the previous ideas.
In particular, temporal extended measurements modelled by using a multi-layered slit experiment.   These and other applications will be discussed in forthcoming papers.



\section{Conclusions and discussion}

 A unified description of Feynman's composition law for amplitudes and Schwinger's transformation functions is provided within the groupoid framework of Quantum Mechanics recently developed in \cite{Ib18,Ib18b}.  An analysis of the statistical interpretation of the formalism is provided using as a fundamental notion the $C^*$ algebra of the groupoid and their states.  Actually, it is shown that any state on the algebra of virtual transitions defines a decoherence functional (by means of the corresponding smeared character) and consequently a grade-2 measure, or a quantum measure in Sorkin's statistical intepretation of quantum mechanics.    Then, either by starting from a quantum measure, or a state on the groupoid algebra, there is a natural notion of amplitudes, called in the text \textit{transition amplitudes}, which subsume the statistical interpretation of the theory.    The groupoids based formalism provides a 'sum-over-histories'-like formula to compute the transition amplitudes and a natural theory of their representations in terms of vector-valued measures.  
 
The states, or decoherence functionals, leading to Feynman's composition law are indentified as idempotent positive semi-definite function on the groupoid and, moreover, a natural factorization condition, isolates those states whose amplitudes satisfy Dirac-Feynamn's principle, that is, they have the form $e^{is}$, with a $s$ an action-like function defined on the groupoid of transitions.    Such states, called \textit{factorizable} in the text, can be given a dynamical interpretation using a dynamical principle for the action function $s$ as in Schwinger's original setting or, alternatively, by using Feynman's construction of the wave function and the corresponding Schr\"odinger's equation.     These ideas will be explored and will constitute the main argument of a forthcoming work.

The work developed in this paper suggests a histories interpretation of the groupoid formalism.   The notion of transition, the abstract Schwinger's notion of selective measurement that changes the state of the system, has a clear dynamical meaning, however, in Schwinger's conceptualisation, such transitions are elementary and not subjected to further scrutiny, while a dynamical description of the change of a system involves an analysis, that is, a decomposition of such change.    This suggest a histories-based approach to the groupoid formalism where the composition of transitions would be interpreted dynamically.   In this sense, the formalism described in the present paper can be understood as a coarse-grained histories interpretation of Schwinger's algebra, where only the sources and targets, i.e., the events of the theory, are selected.    A fine-grained histories description of the theory is needed to provide a proper interpretation of the dynamical nature of the aforementioned factorizable states, and then, of Schwinger's dynamical principle.   As commented before, this will be the objective of another work.


\section*{Acknowledgments}

The authors acknowledge financial support from the Spanish Ministry of Economy and Competitiveness, through the Severo Ochoa Programme for Centres of Excellence in RD (SEV-2015/0554).
AI would like to thank partial support provided by the MINECO research project  MTM2017-84098-P  and QUITEMAD+, S2013/ICE-2801.   GM would like to thank partial financial support provided by the Santander/UC3M Excellence  Chair Program 2019/20.




\begin{thebibliography}{0}

\bibitem{Fe05} R. P. Feynman.  \textit{Feynman's Thesis: A New Approach to Quantum Theory}.   Editor L. M. Brown, World Scientific (2005).  Reprinted from R. P. Feynman, \textit{The principle of least action in Quantum Mechanics} (1942). 

\bibitem{Fe48} R. P. Feynman. \textit{Space-time approach to non-relativistic quantum mechanics}.  Rev. Mod. Phys., \textbf{20}, 367--387 (1948).

\bibitem{Di33} P.A.M. Dirac. The Lagrangian in Quantum Mechanics.
 \textit{Physikalische Zeitschrift der Sovietunion}, Band 3, Heft 1 (1933) 312--320.  Also available as an appendix in 
Feynman's Thesis:  \textit{A New Approach to Quantum Theory}.  Edited by Laurie M.Brown 
World Scientific.  
  
\bibitem{Yo68} W. Yourgrau, S. Mandelstam.  \textit{Variational Principles in Dynamics and Quantum Theory}.  Saunders, Philadelphia, 3rd ed. (1968).
%
\bibitem{Sc91}  J. Schwinger.  \textit{Quantum Kinematics and Dynamics}. Frontiers in Physics, W.A. Benjamin, Inc., (New York 1970); ibid., \textit{Quantum Kinematics and Dynamics}.  Advanced Book Classics,  Westview Press (Perseus Books Group 1991); \textit{ibid.,} \textit{Quantum Mechanics Symbolism of
Atomic measurements,}, edited by  Berthold-Georg Englert. Springer--Verlag, (Berlin 2001).

\bibitem{Sc51}  J. Schwinger.  \emph{The Theory of Quantized Fields I}.  Phys. Rev., {\bf 82} (6) 914-927 (1951); \emph{ibid.} \emph{The Theory of Quantized Fields II}. Phys. Rev., {\bf 91} (3)  713--728 (1953); \emph{ibid.} \emph{The Theory of Quantized Fields III}. Phys. Rev., {\bf 91} (3) 728--740  (1953);  \emph{The Theory of Quantized Fields IV}.  Phys. Rev., {\bf 92} (5) 1283--1300 (1953); \emph{ibid.} \emph{The Theory of Quantized Fields V}. Phys. Rev., {\bf 93} (3) 615--624 (1954); \emph{ibid.} \emph{The Theory of Quantized Fields VI}.  Phys. Rev., {\bf 94} (5)  1362--1384 (1954).

\bibitem{So94} R.D. Sorkin.  \textit{Quantum mechanics as quantum measure theory}. Modern Physics Letters A, \textbf{9}(33), 3119--3127 (1994).

\bibitem{Ib18} F. M. Ciaglia, A. Ibort, G. Marmo. \textit{Schwinger's Picture of Quantum Mechanics I: Groupoids}.  Int. J. Geom. Meth. Modern Phys.,  (2109). arXiv:1905.12274 [math-ph]. DOI: 10.1142/S0219887819501196.  
%
\bibitem{Ib18b}  F. M. Ciaglia, A. Ibort, G. Marmo. \textit{Schwinger's Picture of Quantum Mechanics II: Algebras and observables}.   Int. J. Geom. Meth. Modern Phys.,  (2109). DOI: 10.1142/S0219887819501366.
%
\bibitem{Ib18c}  F. M. Ciaglia, A. Ibort, G. Marmo. \textit{Schwinger's Picture of Quantum Mechanics IV: Composition of systems}. In preparation (2019).

\bibitem{Gu09} S. Gudder. \textit{Quantum measure and integration theory}, J. Math. Phys. \textbf{50} (2009), 123509.

\bibitem{Si09} U. Sinha, C. Couteau, Z. Medendorp, I. S\"ollner, R. Laflamme, R. Sorkin, G. Weihs. \textit{Testing Born's rule in quantum mechanics with a triple slit experiment}. In AIP Conference Proceedings Vol. \textbf{1101} (1), 200--207, AIP (2009).

\bibitem{Ne17} R. Nesselrodt, E. Gagnon, A. Lytle, J. Moreno.  \textit{A New Metric For Triple-Slit Tests of Born's Rule}. In APS April Meeting Abstracts. (2017, January).\verb+ http://meetings.aps.org/link/BAPS.2017.APR.F1.18+

\bibitem{So95} R.D. Sorkin. \textit{Quantum measure theory and its interpretation}. arXiv preprint gr-qc/9507057 (1995). 


\bibitem{So16} A. M. Frauca, R. Sorkin. How to measure the quantum measure. \textit{International Journal of Theoretical Physics}, \textbf{56} (1) (2017) 232--258.
%
\bibitem{Ib18d} A. Ibort, M. A. Rodr\'{\i}guez.  On the structure of finite groupoids and their representations. \textit{Symmetry}, \textbf{11}, 414 (2019).
%
\bibitem{Re80} J.~Renault: \textit{\ A Groupoid Approach to }$C^{\star }-$%
\textit{Algebras}, Lect. Notes in Math. 793, Springer-Verlag (Berlin 1980).
%


\bibitem{Co19} F. Di Cosmo, A. Ibort, G. Marmo.  \textit{Groupoids and coherent states}. Open systems and information dynamics, \textbf{26}(4),  (2019).
%
\bibitem{Ge90} M. Gell-Mann, J.B. Hartle:  \textit{Quantum Mechanics in the Light of
Quantum Cosmology}, in Complexity, Entropy, and the Physics of Information, W.H. Zurek, ed., Reading, Addison-Wesley, 425--59 (1990); \textit{ibid.}, \textit{Classical Equations for Quantum Systems}, Physical Review D, \textbf{47}, 3345--82 (1993).

\bibitem{Ib14} A. Ibort, P. Linares, J.L.G. Llavona. \textit{On the multilinear Hausdorff problem of moments}. 
Revista Matem\'atica Complutense, {\bf 27} (1) 213--224 (2014). 

\bibitem{Do10} H. F. Dowker, S. Johnston, R. Sorkin, \textit{Hilbert Spaces from Path Integrals}, J. Phys. A, \textbf{43}, 275302 (2010), [arXiv:1002.0589].

\bibitem{Gu12} S. Gudder. \textit{Hilbert space representations of decoherence functionals and quantum measures}, Mathematica Slovaca. \textbf{62}, 1209--1230 (2012).

\bibitem{Ib19} A. Ibort, M.A. Rodr\'{\i}guez.  \textit{An introduction to the theory of groups, groupoids and their representations}, CRC (2019).

\bibitem{Do10b} F. Dowker, S. Johnston, S. Surya.  \textit{On extending the quantum measure}.  J. Phys. A: Math. Theor. \textbf{43} 505305 (20pp) (2010).

\bibitem{Ka51} R. V. Kadison, \textit{A representation theory for commutative topological algebra}, Mem. Amer. Math. Soc. No. \textbf{7} (1951) 39 pp.

\bibitem{Fa13} F. Falceto, L. Ferro, A. Ibort, G. Marmo.  \textit{Reduction of Lie-Jordan Banach algebras and quantum states}. J. Phys. A: Math. Theor., \textbf{48} 015201 (14pp) (2013).

\bibitem{AlfShu} E. M. Alfsen, F. W. Shultz. \textit{State spaces of operator algebras}. Springer (2001).

\bibitem{AlfShu2} E. M. Alfsen, F. W. Shultz. \textit{Geometry of state spaces of operator algebras}. Birkh\"{a}user (2003).

\bibitem{La98} N.P. Landsman.  \textit{Mathematical Topics between Classical and Quantum Mechanics}. Springer (1998).

\bibitem{Ga09} Y.  Ghazi-Tabatabai.  \textit{Quantum measure theory: A new interpretation}.  Ph.D. Thesis, Imperial College (2009).


\end{thebibliography}
\end{document}